\documentclass{article}

\usepackage[utf8]{inputenc}
\usepackage{hyperref}
\usepackage{booktabs}
\usepackage{amsfonts}
\usepackage{amsmath,amssymb,amsthm}
\usepackage{nicefrac}
\usepackage{algorithm}
\usepackage{algpseudocode}
\usepackage{graphicx}
\usepackage{subcaption}
\usepackage{multirow}
\usepackage{enumitem}
\usepackage{xcolor}
\usepackage{tikz}
\usepackage{natbib}
\usepackage{authblk}
\usepackage{bbm}
\usepackage{fancyhdr}
\usepackage{extramarks}
\usepackage{calc}  
\usepackage{ltablex}
\usepackage{commath}
\usepackage{mathtools}
\usepackage{bm}

\hypersetup{
    colorlinks=true,
    linkcolor=blue,
    filecolor=magenta,      
    urlcolor=cyan,
    citecolor = blue,
    pdftitle={Overleaf Example},
    pdfpagemode=FullScreen,
    }

\newtheorem*{assumption*}{\assumptionnumber}
\providecommand{\assumptionnumber}{}
\makeatletter
\newenvironment{assumption}[1]
 {%
  \renewcommand{\assumptionnumber}{Assumption #1}%
  \begin{assumption*}%
  \protected@edef\@currentlabel{#1}%
 }
 {%
  \end{assumption*}
 }
\makeatother

\newtheorem{prop}{Proposition}
\newtheorem{thm}{Theorem}
\newtheorem{lem}{Lemma}

\newtheorem{Ex}{Example}

\newcommand{\E}{\mathbb{E}}

\title{
    \textmd{\textbf{A Single-Sample Polylogarithmic Regret Bound for Nonstationary Online Linear Programming}}\\
    }
\author[$\dagger$]{Haoran Xu\thanks{Corresponding Author: haoran14@stanford.edu}}
\author[$\ddagger$]{Owen Shen}

\author[$\dagger$]{Peter Glynn}
\author[$\dagger$]{Yinyu Ye}
\author[$\ddagger$]{Patrick Jaillet}

\affil[$\dagger$]{Department of Management Science and Engineering, Stanford University}
\affil[$\ddagger$]{Operations Research Center, Massachusetts Institute of Technology}
\date{}

\begin{document}
\maketitle
\begin{abstract}
    We study nonstationary Online Linear Programming (OLP), where $n$ orders arrive sequentially with reward-resource consumption pairs that form a sequence of independent, but not necessarily identically distributed, random vectors. At the beginning of the planning horizon, the decision-maker is provided with a resource endowment that is sufficient to fulfill a significant portion of the requests. The decision-maker seeks to maximize the expected total reward by making immediate and irrevocable acceptance or rejection decisions for each order, subject to this resource endowment. We focus on the challenging single-sample setting, where only one sample from each of the $n$ distributions is available at the start of the planning horizon.

    We propose a novel re-solving algorithm that integrates a dynamic programming perspective with the dual-based frameworks traditionally employed in stationary environments. In the large-resource regime—where the resource endowment scales linearly with the number of orders—we prove that our algorithm achieves $O((\log n)^2)$ regret across a broad class of nonstationary distribution sequences. Our results demonstrate that polylogarithmic regret is attainable even under significant environmental shifts and minimal data availability, bridging the gap between stationary OLP and more volatile real-world resource allocation problems. 
\end{abstract}

\section{Introduction}
Resource allocation is a core problem in Operations Research. It focuses on allocating limited resources to maximize overall benefit. In many real-world applications, these decisions must be made in an ``online" fashion. The decision-maker must commit to an action immediately as each request arrives, and the decisions are made without knowing what the future holds. When the resource allocation problem is linear, Online Linear Programming (OLP) (\cite{Agrawal2014}) is a highly effective framework for solving these problems. Its applications on real-world problems such as online advertisement and revenue management have been extensively studied in the literature (\cite{li2022online}).

While standard models of OLP often assume a stable environment, real-world problems are frequently nonstationary. This means the patterns of arrivals and the value of requests can change over time due to seasonal patterns or long-term trends (\cite{Jiang2025}). Ignoring these changes can lead to poor performance, such as running out of resources too early or missing high-value opportunities. Thus, developing OLP algorithms that can adapt to these nonstationary environments is critical. These algorithms must maintain strong performance over different kinds of nonstationary environment and can be implemented when data is scarce and limited historical information is available.

In light of these challenges, this paper seeks to address nonstationary OLP problems. Specifically, we investigate the case that the reward-resource consumption pairs follow a sequence of independent but not necessarily identical distributions. We focus on a minimal data availability situation in which only one sample from each of these distributions are available at the start of the planning horizon. The main research question of this paper is: can we design a near-optimal algorithm for nonstationary OLP under this single-sample setting?
\subsection{Main Contribution}
In this work, we consider a nonstationary OLP framework involving $n$ sequentially arriving orders and a fixed resource endowment provided at the start of the planning horizon. The decision-maker must make immediate and irrevocable acceptance or rejection decisions for each order to maximize the total expected reward, subject to the available resource endowment. We allow the marginal distributions of both rewards and resource consumptions to be independent but non-identical. We assume the marginal distributions of resource consumption share a common support—which may be discrete or continuous—while the conditional distribution of rewards, given resource consumption, is required to be continuous. Notably, we study the single-sample setting: at the start of the horizon, the decision-maker has access to only one sample from each of the $n$ reward-resource consumption distributions. Furthermore, we assume a large-resource regime where the resource endowment scales linearly with the number of orders.

Adopting a dynamic programming perspective, we propose a dual-based pricing algorithm for nonstationary OLP. We define the system state as the average remaining resource. Upon the arrival of each order, the algorithm re-solves a linear program formulated using the current system state and the provided samples of both current and future orders. The optimal dual solution of this linear program yields a price that serves as a threshold, and an order is accepted if its reward exceeds this threshold. We measure the algorithm's performance using regret, defined as the difference between the expected optimal total reward of the offline problem and the expected reward earned by our algorithm. Under a set of boundedness and smoothness assumptions on the nonstationarity, we prove that our algorithm achieves $O(\log^2 n)$ regret for any distribution sequence satisfying these assumptions.

Our regret analysis centers on two pillars: establishing a uniform dual convergence result for the algorithm's prices and characterizing the evolution of the system state. We first demonstrate the existence of a near-optimal offline policy that mirrors our algorithm's structure but utilizes a constant price throughout the planning horizon. We then derive our regret bounds by proving that the prices generated by our algorithm closely track this fixed ``offline price". A key theoretical finding of this analysis is that, despite the nonstationary arrivals and rewards, the prices utilized by the algorithm remain relatively stationary over the planning horizon.

Crucially, while the prices are stable, the system states themselves evolve nonstationarily. To analyze the relationship between the prices and the system states, we introduce the concept of $\delta$-path, defined as a sequence of points in the state space over the planning horizon. We first construct a neighborhood around each point along the $\delta$-path. Our uniform dual convergence result than establishes that if the system state lies within the neighborhood of the $j$-th point of the $\delta$-path when order $j$ arrives, the mean squared error between the computed price and the offline price is $O(\frac{\log(n-j)}{n-j})$ uniformly over the neighborhood. Furthermore, we show that the system state remains within these neighborhoods until the final stages of the planning horizon. Together, these results yield the $O(\log^2 n)$ regret upper bound.
\subsection{Related Works}
Literature on Online Linear Programming (OLP) generally focuses on two primary models: the random permutation model (\cite{Agrawal2014}) and the stochastic input model (\cite{li2022online}). This paper adopts the stochastic input model. Most research under this model assumes that the rewards and the resource consumptions follow independent and identically distributed (i.i.d.) processes. Within the i.i.d. setting, two major distributional assumptions prevail. The first assumes the joint distribution of the reward and the resource consumption has finite support, which typically leads to constant regret (\cite{Chen.et.al.2023}, \cite{Vera2021}). The second assumes the conditional distribution of reward, given the resource consumption, is continuous, which generally yields logarithmic regret (\cite{li2022online}, \cite{Bray2024}). Our work focuses on the setting of continuous conditional distributions. Under this setting, when the available resources grow proportionally with the number of orders, \cite{li2022online} proposed the Action-History-Dependent Learning Algorithm (AhdLA) and established an $O(\log n \log \log n)$ regret bound. \cite{Bray2024} studied the same algorithm under slightly broader assumptions, establishing an $\Omega(\log n)$ regret lower bound and refining the dual convergence analysis to prove a tighter $O(\log n)$ upper bound. In this paper, we further generalize the problem setting by relaxing the i.i.d. requirement to consider distributions that are independent but non-identical.

AhdLA belongs to the class of dual-based ``re-solving" algorithms, which compute resource prices by solving a dual linear program at each time period. These prices determine the thresholds used to accept or reject incoming orders. While these algorithms can achieve optimal regret (\cite{Bray2024}), they are often computationally demanding. As an alternative, first-order methods compute prices via iterative gradient descent. While \cite{lisunye2023} and \cite{balseiro2020dual} establish $O(\sqrt{n})$ regret for such methods, \cite{balseiro2020dual} further demonstrates that this performance holds for nonstationary inputs that are ergodic or seasonal. \cite{Jiang2025} explored nonstationary inputs without requiring ergodic or seasonal properties, assuming instead that rewards and consumptions are independent but non-identical. By introducing a ``deviation budget" ($W_n$) to measure the accuracy of prior estimates of the distributions, they showed that their Informative Gradient Descent Algorithm achieves $O(\max\{\sqrt{n}, W_n\})$ regret. \cite{CheungLyu2025} proposes a algorithmic design framework for more general nonstationary online planning problem. When this framework is applied to the setting in \cite{Jiang2025}, fewer samples are required to construct good prior estimates of the distributions to achieve $O(\sqrt{n})$ regret.

Moving to an even more constrained data environment, \cite{Ghuge2025} investigated the ``single-sample" setting, where only one sample from each distribution is available. They proved that a $1-\epsilon$ approximation ratio is achievable, provided the resource endowment is $\tilde\Omega(1/\epsilon^6)$. Our paper also operates in the single-sample setting but leverages additional structural assumptions to derive superior bounds. We demonstrate that $O(\log^2 n)$ regret—equivalent to a $1 - O(\frac{\log^2 n}{n})$ approximation ratio—can be achieved with standard $\Theta(n)$ resources. This work further distinguishes itself from \cite{Feng2026}, who also assume a single-sample setting but restrict resource consumption to finite support and require conditional reward distributions to remain identical over time. In contrast, our model allows both the conditional reward distributions to be non-identical and the resource consumption support to be infinite, while achieving a slightly improved regret upper bound of $O(\log^2 n)$.

In addition, the concept of variation budget introduced in \cite{Besbes2015} for general stochastic optimization is used to develop conditions under which poly-algorithm can be achieved in our work. Furthermore, our work is also related to papers studying other classes of stochastic optimization problem in nonstationary environment, such as Bandit with Knapsack (\cite{Besbes2014},\cite{Liu2022},\cite{Zhang2024}), Reinforcement Learning (\cite{Lykouris2025}), and online allocation of reusable resources (\cite{Zhang2025}).

The paper is organized as follows. We introduce the problem formulation, assumptions and the algorithm in Section \ref{Section: Problem Setup and Algorithm}. The concept of $\delta$-path and the uniform dual convergence result is given in Section \ref{Section: Dual Convergence}. The analysis of the average remaining resource process is given in Section \ref{Section: RemainResource}. The regret upper bound is provided in Section \ref{Section: Regret Analysis}. Finally, we conclude our work and provide future extensions in Section \ref{Section: Conclusion}. Proofs of all the lemmas, propositions and theorems are in Section \ref{Section: Appendix}.
\section{Problem Setup and Algorithm}\label{Section: Problem Setup and Algorithm}. 
A linear resource allocation problem has the following linear integer programming (LIP) formulation.
\begin{equation}\label{eqn: OfflinePrimalLIP}
    \begin{array}{lll}
    \max & \sum_{j=1}^n u_{j,n}x_{j,n}    &  \\
    S.T  & \sum_{j=1}^n a_{j,n}x_{j,n}\leq b_{0,n}&\\
    & x_{j,n}\in \{0,1\} & \forall 1\leq j\leq n
    \end{array}
\end{equation}
where $n$ is the total number of orders, and $m$ is the number of types of resources. $b_{0,n}\in\mathbb{R}_{+}^m$ is the total amount of available resources, and  $(u_{j,n},a_{j,n})\in\mathbb{R}_{+}\times\mathbb{R}_{+}^{m}$ is the pair of the reward and the resource consumption of the $j$th order of these $n$ orders. $x_{j,n}$ is the decision variable which equals $1$ if the decision maker (DM) accepts order $j$ and equals $0$ otherwise. In the problem of OLP (\cite{Agrawal2014}), the DM needs to solve the above linear resource allocation problem when the $n$ orders arrive sequentially. Once an order arrived, the DM needs to make immediate and irrevocable decision. In this paper, we model $\{(u_{j,n},a_{j,n})\}_{j=1}^n$ as a stochastic process defined on some probability space $(\Omega,\mathcal{F},\mathbb{P})$. Let $\mathcal{F}_{0}=\{\emptyset,\Omega\}$ and $\mathcal{F}_j$ be the sigma algebra generated by $\{(u_{k,n},a_{k,n})\}_{k=1}^j$. Denote the algorithm used by the DM as $\pi$ and the decision on the $j$th order of the $n$ orders under $\pi$ as $x_{j,n}^{\pi}$. Then, an algorithm $\pi$ is feasible in OLP if $x_{j,n}^{\pi}$ is measurable with respect to $\mathcal{F}_t$ and $\sum_{j=1}^na_{j,n}x_{j,n}^{\pi}\leq b_{0,n}$ almost surely. Define the total reward earned by an algorithm $\pi$ to be
\begin{equation*}
    R_n(\pi) = \sum_{j=1}^n u_{j,n}x_{j,n}^{\pi}
\end{equation*} 
The goal of the DM is to maximize the expected total reward $\E[R_n(\pi)]$ over all feasible policy $\pi$. 

We call the following linear relaxation of (\ref{eqn: OfflinePrimalLIP}) the offline LP.
\begin{equation}\label{eqn: OfflinePrimalLP}
    \begin{array}{lll}
    \max & \sum_{j=1}^n u_{j,n}x_{j,n}    &  \\
    S.T  & \sum_{j=1}^n a_{j,n}x_{j,n}\leq b_{0,n}&\\
    & 0\leq x_{j,n}\leq 1 & \forall 1\leq j\leq n
    \end{array}
\end{equation}
and denotes its optimal objective value as $R^{off}_n$. We consider nonstationary OLP problem in the sense that the stochastic process $\{(u_{j,n},a_{j,n})\}_{j=1}^n$ is nonstationary. Let $P=\{P_n\}_{n=1}^{+\infty}$ be a sequence of distributions, where $P_n$ is the joint distribution of $\{(u_{j,n},a_{j,n})\}_{j=1}^n$ when there are $n$ orders in total. When there are $n$ orders in total, the regret of an algorithm $\pi$ under $P$ is defined to be
\begin{equation*}
    Regret_n(\pi,P) = \mathbb{E}_{P_n}\left[R^{off}_n - R_n(\pi)\right]
\end{equation*}
In this paper, we consider the nonstationary OLP problems under the ``large resource'' scenario. Let $\mathbb{R}_{+}^m$ be the set of vectors in $\mathbb{R}^m$ with all entries positive.
\begin{assumption}{1}\label{Assumption: Asymptotic Regime}
    There exists $d_0\in\mathbb{R}^m_{+}$ such that $b_{0,n}=n\cdot d_0$.
\end{assumption}
We call $d_0$ the initial average resource, and we want to study the performance of an algorithm when $n$ goes to infinity with $d_0$ fixed. When nonstationarity is considered in OLP problems, an algorithm that performs well over a reasonably large class of distributions of the nonstationary process $\{(u_{j,n},a_{j,n})\}_{j=1}^n$ is usually desired. Let $\mathcal{P}$ be the family of sequences of distributions that we are interested in. The goal of this paper is to find appropriate regularity conditions on $\mathcal{P}$ and construct an algorithm $\pi$ under which there exists a constant $C$ such that for all $P\in\mathcal{P}$, when $n$ is large enough,
\begin{equation*}
    Regret_n(\pi,P)\leq C (\log n)^2
\end{equation*} 
\subsection{Single-sample Setting}
We certainly prefer less restrictive regularity conditions so that the proposed algorithm works on larger $\mathcal{P}$; however, how strong the regularity conditions are depends on how much we know about the nonstationarity. When the distribution of $\{(u_{j,n},a_{j,n})\}_{j=1}^n$ is completely unknown for each $n$, we provide the following example shows that, even if there are only two sequences of distributions in $\mathcal{P}$ and each of the sequence of distributions is uniformly bounded, there is no feasible online algorithm that has sublinear regret for all $P\in\mathcal{P}$ 
\begin{Ex}\label{Example: LinearRegret}
    Let $\mathcal{P} = \{P_{1},P_{2}\}$, where $P_1=\{P_{1,n}\}_{n=1}^{+\infty}$ and $P_2=\{P_{2,n}\}_{n=1}^{+\infty}$. Under both $P_{1,n}$ and $P_{2,n}$, $a_{j,n}$'s are i.i.d., and $a_{1,n}=1$ almost surely. Under $P_{1,n}$, $\{u_{j,n}\}_{j=1}^n$ a sequence of i.i.d. $Uniform(0,1)$ random variables. Under $P_{2,n}$, $\{u_{j,n}\}_{j=1}^n$ is a sequence of independent random variables such that $u_{j,n}\sim Uniform(0,1)$ if $j\leq \left\lceil\frac{n}{2}\right\rceil$, and $u_{j,n}\sim Uniform(2,3)$ otherwise. The initial average resource is $d_0=\frac{1}{4}$. Then, $\max\{Regret_n(\pi,P_1),Regret_n(\pi,P_2)\}$ is $\Omega(n)$ for any feasible online algorithm $\pi$.
\end{Ex}
Thus, to achieve polylogarithmic regret for each sequences of distributions in a reasonably large distribution classes $\mathcal{P}$, some information regarding the nonstationarity is required. In this paper, we focus on the following \textit{single-sample} setting.
\begin{assumption}{2}\label{Assumption: Single-sample}
    $\{\tilde u_{j,n}, \tilde a_{j,n}\}_{j=1}^n$ is available at the beginning of the planning horizon such that (i) $\{\tilde u_{j,n}, \tilde a_{j,n}\}_{j=1}^n$ and $\{u_{j,n}, a_{j,n}\}_{j=1}^n$ are independent. (ii) $\{\tilde u_{j,n}, \tilde a_{j,n}\}_{j=1}^n$ and $\{u_{j,n}, a_{j,n}\}_{j=1}^n$ are equal in distribution.
\end{assumption}
\subsection{Algorithm}
We can take a dynamic programming perspective on the offline LP (\ref{eqn: OfflinePrimalLP}). Define
\begin{equation}\label{eqn: OfflineDP}
    V_{j,n}(b) =\left\{
        \begin{array}{ll}
            \max u_{j+1,n}x_{j+1,n} + V_{j+1,n}(b-a_{j+1,n}x_{j+1,n})\quad\mathrm{s.t\;} a_{j,n}x_{j,n}\leq b,\; 0\leq x_{j,n}\leq 1 & \qquad\mathrm{if}\;j<n \\
           0  &\qquad \mathrm{if}\;j = n 
        \end{array}
    \right.   
\end{equation}
Then, $V_{0,n}(b_{0,n})=R_n^{off}$. In addition, for all $0\leq j< n$
\begin{equation*}
    \begin{array}{llll}
        V_{j,n}(b) = &\max & \sum_{k=j+1}^n u_{k,n}x_{k,n}    &  \\
        &S.T  & \sum_{k=j+1}^n a_{k,n}x_{k,n}\leq b&\\
        && 0\leq x_{k,n}\leq 1 & \forall j+1\leq k\leq n
    \end{array}
\end{equation*}
After normalizing the objective function by dividing the total number of orders left, the dual form of the above linear program is
\begin{equation*}
    \begin{array}{llll}
        \min & \frac{b^Tp}{n-j} + \frac{1}{n-j}\sum_{k=j+1}^n y_{k,n}    &  \\
        S.T  & u_{k,n} - {a_{k,n}}^Tp\leq y_{k,n} & \forall j+1\leq k\leq n\\
        & p\geq 0,\; y_{k,n}\geq 0& \forall j+1\leq k\leq n
    \end{array}
\end{equation*}
Let $p_{j,n}\left(\frac{b}{n-j}\right)$ be the optimal solution of this dual problem, and let $x^{DP}_{j+1,n}(b)$ be given by evaluating the optimality equation (\ref{eqn: OfflineDP}). By complementary slackness,
\begin{equation*}
    x^{DP}_{j+1,n}(b) = \left\{
    \begin{array}{ll}
       1  &  \mathrm{if}\; u_{j+1,n}>{a_{j+1,n}}^Tp_{j,n}\left(\frac{b}{n-j}\right)\\
       0  &  \mathrm{if}\; u_{j+1,n}<{a_{j+1,n}}^Tp_{j,n}\left(\frac{b}{n-j}\right)
    \end{array}
    \right.
\end{equation*}
When the case $u_{j+1,n}={a_{j+1,n}}^Tp_{j,n}\left(\frac{b}{n-j}\right)$ happens rarely, the above dynamic programming perspective motivates a natural dual-based algorithm $\pi$. Before the arrival of order $j+1$, let $b_{j,n}$ be the amount of remaining resources and we compute an estimator $p^{\pi}_{j,n}$ of the dynamic programming price $p_{j,n}\left(\frac{b_{j,n}}{n-j}\right)$. When order $j+1$ arrives, we then accept it if $u_{j,n}>{a_{j,n}}^Tp^{\pi}_{j,n}$. To construct a good estimator of $p_{j,n}\left(\frac{b_{j,n}}{n-j}\right)$, we rewrite the dual LP above as
\begin{equation*}
    \min_{p\geq 0} \frac{b^Tp}{n-j} + \frac{1}{n-j}\sum_{k=j+1}^n \left(u_{j,n} - {a_{j,n}}^Tp\right)^+
\end{equation*}
Then, it can be viewed as a sample average approximation of the following problem with a single sample for the distributions of the reward and the resource consumption of each of the remaining order.
\begin{equation*}
    \min_{p\geq 0} \frac{b^Tp}{n-j} + \frac{1}{n-j}\sum_{k=j+1}^n \E_{P_n}\left[\left(u_{j,n} - {a_{j,n}}^Tp\right)^+\right]
\end{equation*}
Although we do not know $P_n$, we can use the samples given in Assumption \ref{Assumption: Single-sample} to form another single-sample sample average approximation problem, which provides us an estimator that can be used in the aforementioned dual-based algorithm. Details of the algorithm are given in Algorithm \ref{alg: Simulater}. 
\begin{algorithm}[H]
    \caption{}\label{alg: Simulater}
    \begin{algorithmic}[1]
        \State Input the initial average resource $d_0$ and a random sample $\left\{(\tilde u_{j,n},\tilde a_{j,n})\right\}_{j=1}^n$.   
        \State Set $b_{0,n}=nd_0$
        \For{$j=1,\cdots,n$}
            \State \begin{equation*}
                p^{\pi}_{j-1,n} = \arg\min_{p\geq 0} {b_{j-1,n}}^Tp + \sum_{k=j}^{n}\left(\tilde u_{k,n}-{\tilde a_{k,n}}^Tp\right)^+
            \end{equation*}
            \State $x_{j,n}^{\pi}=\mathbbm{1}\left\{u_{j,n}>{a_{j,n}}^{T}p_{j-1,n}^{\pi}\right\}\mathbbm{1}\left\{b_{j-1,n}\geq a_{j,n}\right\}$
            \State $b_{j,n}=b_{j-1,n}-a_{j,n}x_{j,n}^{\pi}$
        \EndFor
    \end{algorithmic}
\end{algorithm}
Algorithm \ref{alg: Simulater} can be viewed as a generalization of the Action-History-Dependent Learning Algorithm (AhdLA) studied in \cite{li2022online} and \cite{Bray2024} to the nonstationary setting under the single-sample assumption. In addition, from the computation perspective, the computation complexity of Algorithm \ref{alg: Simulater} is the same as AhdLA.
\subsection{Boundedness and Smoothness Assumptions on Distributions Classes}
When we analyze the regret of Algorithm \ref{alg: Simulater}, we focus on the distribution classes satisfying the following boundedness and smoothness assumptions.
\begin{assumption}{3}\label{Assumption: Nonstationarity}
    There exists positive real numbers $\alpha$, $\underline{a}$, $\lambda$, $\bar L$, $\bar u$, $\bar\mu$, $\underline{\mu}$, and $\Xi\subseteq [0,\alpha]^m$ such that for any $P\in\mathcal{P}$,
    \begin{enumerate}[label=(\alph*)]
        \item $\left\{\left(a_{j,n},u_{j,n}\right)\right\}_{j=1}^n$ is a sequence of independent random pair in $\mathbb{R}_{+}\times\mathbb{R}^{m}_{+}$ for all $n$.
        \item $\Xi$ is the support of $a_{j,n}$, the smallest entry of $\E[a_{j,n}]$ is greater than $\underline{a}$, and the smallest eigenvalue of $\mathbb{E}_{P_n}\left[a_{j,n}{a_{j,n}}^T\right]$ is greater than $\lambda$ for all $j$ and $n$.
        \item There exists a positive continuously differentiable function $v:[0,1]\times[0,\alpha]^m\rightarrow [0,\bar \mu]$ such that $\norm{\nabla v}_2\leq \bar L$ and the restriction of $v\left(\frac{j}{n},\cdot\right)$ to $\Xi$ is either the probability mass function or the probability density function of $a_{j,n}$ for all $j$ and $n$.
        \item There exists a continuously differentiable function $f(t,a,u):[0,1]\times[0,\alpha]^m\times[0,\bar u]\rightarrow [\underline{\mu},\bar\mu]$ such that $\norm{\nabla f}_2\leq \bar L$ and $f\left(\frac{j}{n},a,\cdot\right)$ is the conditional density of $u_{j,n}$ given $\left\{a_{j,n}=a\right\}$ for all $a\in\Xi$, $j$, and $n$.
    \end{enumerate}
\end{assumption}
\noindent We will do regret analysis by fixing a sequence of distributions $P\in\mathcal{P}$, and the number of orders is included in the notation of the rewards and the resource consumptions. Thus, to simplify the notation, we will use $\E[\cdot]$ and $P\{\cdot\}$ to denote expectation and probability instead of $\E_{P_n}[\cdot]$ and $P_n\{\cdot\}$.  

Assumption \ref{Assumption: Nonstationarity} (a) relaxes the standard i.i.d. assumptions in the literature to allow us consider nonstationary rewards and resource consumptions. We emphasize that the joint distribution of the reward and the resource consumption of the $j$th order may change with the total number of orders $n$ under Assumption \ref{Assumption: Nonstationarity} (a). Assumption \ref{Assumption: Nonstationarity} (b), (c), (d) are regularity conditions on the boundedness and smoothness of the distribution of $\{u_{j,n}, a_{j,n}\}_{j=1}^n$ required by our regret analysis. We introduce several useful Lemmas and Propositions implied by these assumptions. Define
\begin{equation*}
    g^*_{j,n}(p,d) = d^Tp + \frac{1}{n-j}\sum_{k=j+1}^n\E\left[\left(u_{k,n}-{a_{k,n}}^Tp\right)^+\right]
\end{equation*}
and for a fixed $d$, define
\begin{equation*}
    p^*_{j,n}(d)\in\arg\min_{p\geq 0} g^*_{j,n}(p,d).
\end{equation*}
For simplicity, we write $g^*_{0,n}(p,d)$ as $g^*_n(p,d)$ and $p^*_{0,n}(d)$ as $p^*_n(d)$. As $g^*_{j,n}(\cdot,\cdot)$'s are functions only depend on the distributions of the rewards and the resource consumptions, we call the optimization problems defining $p_{j,n}^*(\cdot)$'s the \textit{population problems} and call $p_{j,n}^*(\cdot)$'s the \textit{population prices}. For now, we denote $p^*_{j,n}(d)$ as an arbitrary element in the optimal solution set and provide a boundedness result in the following Lemma. A uniqueness result will be given later. Let $d^{min}$ be the smallest entry of the vector $d$.  
\begin{lem}\label{Lemma: BoundedPrice-I}
    Under Assumption \ref{Assumption: Nonstationarity}, for any $P\in\mathcal{P}$, if $d^{min}>0$, then $\norm{p_{j,n}^*(d)}_1\leq\frac{\bar u}{d^{min}}$ for all $j$ and $n$.
\end{lem}
Assumption \ref{Assumption: Nonstationarity} also guarantees some convergence properties of the population problems as $n$ goes to infinity. Let $\left\{(u_t,a_t)\right\}_{t\in[0,1]}$ be a collection of random variables such that $\left(u_{\frac{j}{n}},a_{\frac{j}{n}}\right)\stackrel{d}{=}(u_{j,n},a_{j,n})$ for all $j$ and $n$. Define
\begin{equation*}
    g_{\infty}^*(p,d) = d^Tp+\int_{0}^{1}\mathbbm{E}\left[(u_s-{a_s}^Tp)^+\right]ds
\end{equation*}
The following lemma shows a pointwise convergence property of the objective functions of the population problems defining $p^*_n(d)$.
\begin{lem}\label{Lemma: Riemann Integral}
    Under Assumption \ref{Assumption: Nonstationarity}, for any $P\in\mathcal{P}$, $\lim_{n\rightarrow+\infty}g^*_{n}(p,d_0)=g_{\infty}^*(p,d_0)$ for all $p\in\mathbb{R}^m$.
\end{lem}
We can view $g^*_{\infty}(\cdot,\cdot)$ as the objective function of the limiting population problem. The pointwise convergence result given in Lemma \ref{Lemma: Riemann Integral} allows us to make the non-degeneracy assumption in the following way. 
\begin{assumption}{4}\label{Assumption: p_star}
    There exists $\epsilon_{\mathcal{P}}$ such that for any $P\in\mathcal{P}$, there exists $p^*\in\arg\min_{p\geq 0} g^*_{\infty}(p,d_0)$ such that $(\inf_{a\in \Xi} a^Tp^*, \sup_{a\in \Xi} a^Tp^*)\subseteq(\epsilon_{\mathcal{P}},\bar u-\epsilon_\mathcal{P})$.
\end{assumption}
Note that Assumption \ref{Assumption: p_star} implicitly assumes that $p^*>0$. Together with Assumption \ref{Assumption: Nonstationarity}, we have the following Proposition that summarizes useful technical results used in the regret analysis.
\begin{prop}\label{Proposition: Boundedness and Smoothness}
    Under Assumptions \ref{Assumption: Asymptotic Regime}, \ref{Assumption: Nonstationarity} and \ref{Assumption: p_star}, there exists positive real numbers $\kappa<1$, $\epsilon$, $\epsilon_p$, $\underline{\sigma}$, $\bar\sigma$, $L$ such that for all $P\in\mathcal{P}$.
    \begin{enumerate}[label=(\alph*)]
        \item $\E\left[\left(u_{j,n}-{a_{j,n}}^Tp\right)^+\right]$ is three times continuously differentiable and the absolute value of its third derivatives are bounded above by $L$ on $B_{2\epsilon}(p^*)\equiv\Omega_p\subseteq\mathbb{R}^m_{+}$ for all $j$ and $n$.
        \item On $\Omega_p$, the eigenvalues of the Hessian of $\E\left[\left(u_{j,n}-{a_{j,n}}^Tp\right)^+\right]$ are in $[0,\bar\sigma]$ and the eigenvalues of the Hessian of $g^*_{j,n}(p,d_0)$ are in $[\underline{\sigma},\bar\sigma]$ for all $j$ and $n$.
        \item There exists $N_0$ such that, when $n\geq N_0$, $p_n^*(d_0)$ is uniquely determined, $B_{\epsilon}(p_n^*(d_0))\subseteq\Omega_p$, and each entry of $\E[a_{j,n}\mathbbm{1}\{u_{j,n}>{a_{j,n}}^Tp_n^*(d_0)\}]$ is greater than $\epsilon_p$ if $j\geq\kappa n$.
        \item Let $\underline{s}=\inf_{a\in\Xi,p\in\Omega_p}a^Tp$, $\bar{s}=\sup_{a\in\Xi,p\in\Omega_p}a^Tp$, and
        \begin{equation*}
        L_{j,n} = \frac{\bar\mu\alpha^{m+2}\sum_{k=j+2}^{n}\left(\sup_{a\in\Xi}\sup_{s\in [\underline{s},\bar s]}\abs{f\left(\frac{k}{n},a,s\right)-f\left(\frac{j+1}{n},a,s\right)}+\sup_{a\in\Xi}\abs{v\left(\frac{k}{n},a\right)-v\left(\frac{j+1}{n},a\right)}\right)}{\underline{\sigma}(n-j)(n-j-1)}
    \end{equation*}
        then for all $n\in\mathbb{N}_{+}$,
            \begin{equation*}
                \
                \prod_{j=0}^{n-2}\left(1+L_{j,n}\right)\leq e^{\frac{\mu\bar L\alpha^{m+2}}{\underline{\sigma}}}
            \end{equation*}
    \end{enumerate}
\end{prop}
\noindent 
Proposition \ref{Proposition: Boundedness and Smoothness} (a), (b) and (c) gives the smoothness and strong convexity of $g_{j,n}(p,d)$ for $p$ in a neighborhood of $p_n^*(d_0)$, and Proposition \ref{Proposition: Boundedness and Smoothness} (d) can be viewed as a variation budget requirement. We can view Proposition \ref{Proposition: Boundedness and Smoothness} as more general assumptions to make the regret bound in this paper hold. In fact, our regret analysis holds if Assumption \ref{Assumption: Asymptotic Regime}, Assumption \ref{Assumption: Single-sample}, independence and boundedness of the distributions in Assumption \ref{Assumption: Nonstationarity} and Proposition \ref{Proposition: Boundedness and Smoothness} hold. Thus, although the $\mathcal{P}$ in Example \ref{Example: LinearRegret} does not exactly satisfies Assumption \ref{Assumption: Nonstationarity}, the regret bound in this paper is still true for this $\mathcal{P}$ because Proposition \ref{Proposition: Boundedness and Smoothness} holds under it.

\section{Dual Convergence}\label{Section: Dual Convergence}
In this section, we focus on the error of using $p^{\pi}_{j,n}$ to estimate $p_{j,n}\left(\frac{b_{j,n}}{n-j}\right)$. Define
\begin{equation*}
    g_{j,n}(p,d) = d^Tp + \frac{1}{n-j}\sum_{k=j+1}^n\left(u_{k,n}-{a_{k,n}}^Tp\right)^+
\end{equation*}
Then, $p_{j,n}(d)\in\arg\min_{p\geq 0} g_{j,n}(p;d)$. Note that both $p^{\pi}_{j,n}$ and $p_{j,n}(d)$ may not be unique; however, the selection of these prices from the corresponding optimal solution sets does not affect the regret analysis. Since $p^{\pi}_{j,n}$ and $p_{j,n}(\frac{b_{j,n}}{n-j})$ have the same distribution, and they are both a sample average approximation of the population price $p^*_{j,n}\left(\frac{b_{j,n}}{n-j}\right)$, we choose to focus on the mean squared error between $p_{j,n}(\frac{b_{j,n}}{n-j})$ and $p^*_{j,n}\left(\frac{b_{j,n}}{n-j}\right)$. Similar to the analysis of the i.i.d case as in \cite{Bray2024}, to deal with the randomness of $b_{j,n}$, our goal in this section is to prove a uniform dual convergence result by finding an upper bound of
\begin{equation}\label{eqn: UniformMSE}
    \E\left[\sup_{d\in\Omega_{j,n}}\norm{p_{j,n}(d)-p^*_{j,n}(d)}_2^2\right]
\end{equation}
An interesting question to ask is: how should we select the set $\Omega_{j,n}$? If it is selected too small, then $\frac{b_{j,n}}{n-j}$ cannot fall into it with high probability which will make an upper bound of (\ref{eqn: UniformMSE}) useless. If it is selected too large, (\ref{eqn: UniformMSE}) will be too large to lead us to a polylogarithmic regret upper bound. We answer this question by introducing the concept of \textit{$\delta$-path} in the following subsection.  
\subsection{$\delta$-path}
Define $d^{\pi}_{j,n} = b^{\pi}_{j,n}/(n-j)$. Let $\{x^{off}_{j,n}\}_{j=1}^n$ be the optimal solution of the offline LP (\ref{eqn: OfflinePrimalLP}) and define $d^{off}_{j,n} = (b_{0,n} - \sum_{k=1}^{j}a_{k,n}x^{off}_{k,n})/(n-j)$. If Algorithm \ref{alg: Simulater} has low regret, it is intuitive that the two sequences $\{d^{\pi}_{j,n}\}_{j=1}^n$ and $\{d^{off}_{j,n}\}_{j=1}^n$ are close with high probability. Thus, we may want $\Omega_{j,n}$ to be a neighborhood of $d^{off}_{j,n}$. However, $d^{off}_{j,n}$ is random which makes analyze the mean squared error over its neighborhood difficult. We instead construct a deterministic approximation of $d^{off}_{j,n}$ and let $\Omega_{j,n}$ be the neighborhood of the deterministic approximation. 

Recall that $d_0$ is the initial average resource given in Assumption \ref{Assumption: Asymptotic Regime}. By complementary slackness 
\begin{equation*}
    x^{off}_{j,n} = \left\{
    \begin{array}{ll}
       1  &  \mathrm{if}\; u_{j,n}>{a_{j,n}}^Tp_{0,n}(d_0)\\
       0  &  \mathrm{if}\; u_{j,n}<{a_{j,n}}^Tp_{0,n}(d_0)
    \end{array}
    \right.
\end{equation*}
Assumption \ref{Assumption: Nonstationarity} implies that the conditional distribution of $u_{j,n}$ given $a_{j,n}$ is continuous for all $j$ and $n$. Thus, $\{u_{j,n} = {a_{j,n}}^Tp_{0,n}(d_0)\}$ happens for at most $m$ orders. Thus, when $n>>m$, $\tilde d^{off}_{j,n} = (b_{0,n} - \sum_{k=1}^{j}a_{k,n}\mathbbm{1}\{u_{k,n}>{a_{k,n}}^Tp_{0,n}(d_0)\})/(n-j)$ is almost the same as $d^{off}_{j,n}$. We then replace $p_{0,n}(d_0)$ with the population price $p_n^*(d_0)$ and take expectation to get the following definition of $\delta_{j,n}$. Define
\begin{equation*}
    \delta_{j,n}(d)=\frac{nd-\sum_{k=1}^j\mathbb{E}\left[a_{k,n}\mathbbm{1}\{u_{k,n}>{a_{k,n}}^Tp_n^*(d)\}\right]}{n-j}.
\end{equation*} 
$\{\delta_{j,n}(d_0)\}_{j=1}^n$ is then our deterministic approximation of $\{d^{off}_{j,n}\}_{j=1}^n$, and we will make $\Omega_{j,n}$ a neighborhood of $\delta_{j,n}(d_0)$. The question remained is how large $\Omega_{j,n}$ should be. We answer this question by proving some useful properties of $\{\delta_{j,n}\}_{j=1}^n$. 

We start from a more detailed discussion on the relationship between $\delta$-path and population price. Let $\Omega_p$ be given in Proposition \ref{Proposition: Boundedness and Smoothness}. For each $j$, $n$ and $p\in\Omega_p$, denote $G^*_{j,n}(p)$ as the Hessian of $g^*_{j,n}(\cdot,d)$ at $p$ and define
\begin{equation*}
    \bar{\delta}_{j,n}(p)=\frac{1}{n-j}\sum_{k=j+1}^n\E\left[a_{k,n}\mathbbm{1}\left\{u_{k,n}>{a_{k,n}}^Tp\right\}\right]
\end{equation*}
Define $\bar{\Omega}_{j,n} = \left\{\bar{\delta}_{j,n}(p):p\in\Omega_p\right\}$ to be the range of $\bar{\delta}_{j,n}$. For simplicity, we write $\bar{\delta}_{0,n}$ as $\bar{\delta}_n$ and $\bar{\Omega}_{0,n}$ as $\bar{\Omega}_{n}$. The following Lemma relates the population prices and the $\delta$-paths through $\bar\delta_{j,n}$'s.
\begin{lem}\label{Lemma: Unique_Equal_Population_Price}
    Under Assumptions \ref{Assumption: Asymptotic Regime}, \ref{Assumption: Nonstationarity} and \ref{Assumption: p_star}, for any $P\in\mathcal{P}$, when $n$ is large enough, for all $0\leq j<n$,
    \begin{enumerate}[label=(\alph*)]
        \item $p_{j,n}^*(\delta_{j,n}(d))\in\Omega_p$ is uniquely determined and $p_{j,n}^*(\delta_{j,n}(d))=p_n^*(d)$ for all $d\in \bar{\Omega}_n$.
        \item $\bar{\delta}_{j,n}:\Omega_p\rightarrow\bar{\Omega}_{j,n}$ is invertible and an open mapping, $\delta_{j,n}(d_0)\in \bar{\Omega}_{j,n}$, and $\bar\delta_{j,n}^{-1}(d) = p_{j,n}^*(d)$ for all $d\in \bar{\Omega}_{j,n}$.
        \item $p_{j,n}^*(\cdot)$ is continuously differentiable on $\bar\Omega_{j,n}$, and its Jacobian matrix at $d$ is $-G^*_{j,n}(p_{j,n}^*(d))^{-1}$ for all $d\in\bar\Omega_{j,n}$
    \end{enumerate}
\end{lem}
\noindent Lemma \ref{Lemma: Unique_Equal_Population_Price} shows that $\bar\Omega_{j,n}$ is a neighborhood of $\delta_{j,n}(d_0)$. In addition, Proposition \ref{Proposition: Boundedness and Smoothness} and Lemma \ref{Lemma: Unique_Equal_Population_Price} together show that the population prices $p_{j,n}$'s have good smoothness properties over $\bar\Omega_{j,n}$, and $g^*_{j,n}$'s have good smoothness properties over the neighborhood of the population prices evaluated on $\bar\Omega_{j,n}$. These smoothness properties allow us to simplify the uniform dual convergence problem to an easier empirical process problem. However, to achieve this simplification, we need to further narrow down $\bar\Omega_{j,n}$ to get the neighborhood $\Omega_{j,n}$ required in the uniform convergence result.
\subsection{Uniform Convergence}
We prove the uniform convergence result by generalizing the method in \cite{Bray2024} to the nonstationary case. Consider the following subset of $\bar\Omega_{j,n}$. Let $\epsilon$ be given in Proposition \ref{Proposition: Boundedness and Smoothness}, define
\begin{equation*}
    \Omega_{j,n} = \left\{\bar\delta_{j,n}(p):p\in B_{\frac{\epsilon}{2}}(p_n^*(d_0))\right\}.
\end{equation*}
Lemma \ref{Lemma: Unique_Equal_Population_Price} still holds if $\Omega_p$ is replaced by $B_{\frac{\epsilon}{2}}(p_n^*(d_0))$ and $\bar\Omega_{j,n}$ is replaced by $\Omega_{j,n}$. Thus, $\Omega_{j,n}$ is still a neighborhood of $\delta_{j,n}(d_0)$. We start from an upper bound on the tail probability.
\begin{prop}\label{Proposition: TailProbDualConv}
    Under Assumptions \ref{Assumption: Asymptotic Regime}, \ref{Assumption: Nonstationarity} and \ref{Assumption: p_star}, there exists $C_1$, $C_2$ and $\nu<\frac{\epsilon}{2}$ such that for all $P\in\mathcal{P}$, when $n$ is large enough, 
    \begin{equation*}
        \mathbb{P}\left\{\sup_{d\in\Omega_{j,n}}\norm{p_{j,n}(d)-p^*_{j,n}(d)}_2>\nu\right\}\leq C_1 \exp\left(-C_2\nu^2(n-j)\right) \quad\forall 0\leq j<n
    \end{equation*}
\end{prop}
Proposition \ref{Proposition: TailProbDualConv} shows that the probability of $p_{j,n}(d)$ falling outside a small neighborhood of $p_{j,n}^*(d)$ is exponentially small uniformly over $\Omega_{j,n}$. In addition, $p_{j,n}$ is bounded by the following lemma,
\begin{lem}\label{Lemma: BoundedPrice}
    Under Assumption \ref{Assumption: Asymptotic Regime}, \ref{Assumption: Nonstationarity} and \ref{Assumption: p_star}, there exists $\underline{d}>0$ such that for all $P\in\mathcal{P}$, when $n$ is large enough, for all $0\leq j<n$ and $d\in\Omega_{j,n}$, each entry of $d$ is in $(\underline{d},\alpha]$ and $\norm{p_{j,n}^*(d)}_2\leq\frac{\bar u}{\underline{d}}$ almost surely.
\end{lem}
Thus, we only need to focus on the case $\sup_{d\in\Omega_{j,n}}\norm{p_{j,n}(d)-p^*_{j,n}(d)}_2\leq \nu$, which allow us to use the smoothness properties given by Lemma \ref{Lemma: Unique_Equal_Population_Price} and Proposition \ref{Proposition: Boundedness and Smoothness}. The uniform dual convergnece result is stated in the following theorem.
\begin{thm}\label{Theorem: DualConvergence}
Under Assumption \ref{Assumption: Asymptotic Regime}, \ref{Assumption: Nonstationarity}, \ref{Assumption: p_star}, there exists a constant $C_{Dual}$ such that for all $P\in\mathcal{P}$, when $n$ is large enough,
\begin{equation*}
    \E\left[\sup_{d\in\Omega_{j,n}}\norm{p_{j,n}(d)-p^*_{j,n}(d)}_2^2\right]\leq \frac{C_{Dual}\log(n-j)}{n-j}
\end{equation*}
\end{thm}
The proof idea of Proposition \ref{Proposition: TailProbDualConv} and Theorem \ref{Theorem: DualConvergence} is to upper bound $\sup_{d\in\Omega_{j,n}}\norm{p_{j,n}(d)-p^*_{j,n}(d)}_2^2$ by the supremum of the difference between the gradient of $g^*_{j,n}$ and its empirical version over $\Omega_p$. Complete proofs are provided in the appendix.
\section{Average Remaining resource}\label{Section: RemainResource}
Under Assumption \ref{Assumption: Single-sample}, the uniform dual convergence result in Theorem \ref{Theorem: DualConvergence} also holds for the prices $p_{j,n}^{\pi}$ used in Algorithm \ref{alg: Simulater}. However, the uniform dual convergence can be used only when the average remaining resource $d_{j,n} \equiv b_{j,n}/(n-j)$ is in $\Omega_{j,n}$. When $d_{j,n}$ is not in $\Omega_{j,n}$ or we run out of resources, we have very little control over the quality of the decisions made by Algorithm \ref{alg: Simulater}. To prove the regret upper bound, we need to analyze how many times we cannot apply the uniform dual convergence result when we implement Algorithm \ref{alg: Simulater}. Let $b_{j,n}(i)$ be the $i$th entry of $b_{j,n}$. Define
\begin{equation*}
    \bar\tau_n = \inf\{j\geq 1:\exists i\;\mathrm{s.t.}\;b_{j,n}(i)<\alpha\;\mathrm{or}\;d_{j,n}\notin\Omega_{j,n}\}\wedge n
\end{equation*}
Let $\{\mathcal{F}_{j,n}\}_{j=0}^n$ be the natural filtration of $\{(u_{j,n},a_{j,n})\}_{j=1}^n$, i.e., $\mathcal{F}_{0,n}=\{\emptyset,\Omega\}$ and $\mathcal{F}_{j,n}=\sigma(\{(u_{k,n},a_{k,n})\}_{k=1}^j)$ for $1\leq j\leq n$. Then, $\bar\tau_n$ is a stopping time adapted to $\{\mathcal{F}_{j,n}\}_{j=0}^n$. Our goal of this section is to show $\E[n-\bar\tau_n]$ is small. The following Lemma provides a characterization of $\Omega_{j,n}$ which allow us to constructing stopping times that are easier to analyze and smaller than or equal to $\bar\tau_n$.
\begin{lem}\label{Lemma: UniformCylinder}
    Under Assumptions \ref{Assumption: Asymptotic Regime}, \ref{Assumption: Nonstationarity} and \ref{Assumption: p_star}, there exists $\epsilon_d>0$ such that for all $P\in\mathcal{P}$, when $n$ is large enough,
    \begin{equation*}
        B_{\epsilon_d}(\delta_{j,n}(d_0))\subseteq\Omega_{j,n}\quad\forall 0\leq j<n
    \end{equation*}
\end{lem}
\noindent Let $\underline{d}$ be given in Lemma \ref{Lemma: BoundedPrice}. Then, Lemma \ref{Lemma: BoundedPrice} and \ref{Lemma: UniformCylinder} imply that if  
\begin{equation}\label{Definition: stopping_2}
    \iota \geq \left\lceil\frac{\alpha}{\underline{d}}\right\rceil\;\;\mathrm{and}\;\;0<\epsilon_{j,n}<\epsilon_d\;\forall 0\leq j \leq n-  \iota-1
\end{equation}
then 
\begin{equation*}
    \bar\tau_n\geq\tau_n\equiv \inf\left\{1\leq j\leq n-\iota-1:d_{j,n}\notin B_{\epsilon_{j,n}}(\delta_{j,n}(d_0))\right\}\wedge n-\iota\quad\mathrm{a.s.}
\end{equation*}
In the rest of this section, we focus on how to select appropriate $\iota$ and $\{\epsilon_{j,n}\}_{j=1}^n$ so that $\E[n-\tau_n]$ is $O((\log n))$. Let $e\in\mathbb{R}^m$ be the vector that every entry is $1$. Define
\begin{equation*}
    d'_{j,n} = d_{j,n}\mathbbm{1}\{j\leq \tau_n\} + 2\alpha e\mathbbm{1}\{j> \tau_n\}
\end{equation*}
In other words, $d'_{j,n}$ is the average remaining resource up to the stopping time $\tau_n$ and is fixed to a relatively large value $2\alpha e$ after $\tau_n$. In addition, let $X_{j,n} = d'_{j,n}-\delta_{j,n}(d_0)$, which is $d'_{j,n}$ centered by $\delta_{j,n}(d_0)$. Then, the definition of $\tau_n$ together with Lemma \ref{Lemma: BoundedPrice} implies that $P(\tau_n\leq j)\leq P(\norm{X_{j,n}}_2>\epsilon_{j,n})$ for all $j < n-\iota$. Thus, to get an upper bound of $\E[n-\tau_n]$, we can focus on the tail probability of $X_{j,n}$. Define
\begin{equation*}
    \begin{aligned}
        Z_{j,n} = \mathbb{E}\left[X_{j+1,n}\bigg\rvert\mathcal{F}_{j,n}\right]-X_{j,n}
    \end{aligned}
\end{equation*}
Then $M_{j,n} = X_{j,n} - \sum_{k=0}^{j-1}Z_{k,n}$ forms a martingale adapted to $\{\mathcal{F}_{j,n}\}_{j=0}^n$. Note that $d'_{j,n}$, $X_{j,n}$, $Z_{j,n}$ and $M_{j,n}$ all depend on $\iota$ and $\epsilon_{j,n}$'s selected for the stopping time $\tau_n$, but we do not include them into the notation for simplicity. The following lemma gives an upper bound on the tail probability of $M_{j,n}$.
\begin{lem}\label{Lemma: TailProbMartingale}
    Under Assumption \ref{Assumption: Nonstationarity}, for any $P\in\mathcal{P}$, if $\iota$ and $\{\epsilon_{j,n}\}_{j=0}^{n-\iota-1}$ satisfy (\ref{Definition: stopping_2}), then
    \begin{equation*}
        \begin{aligned}
            P\left\{\norm{M_{j,n}}_2\geq\eta\right\}\leq 2m\exp\left(-\frac{\eta^2(n-j-1)}{(8\alpha+2\epsilon_d)^2}\right)
        \end{aligned}
    \end{equation*}
    for all $n\in\mathbb{N}_+$, $0\leq j\leq n$ and $\eta>0$.
\end{lem}
The fast decay rate of the tail probability of $M_{j,n}$'s given in Lemma \ref{Lemma: TailProbMartingale} suggests that, if we can select $\epsilon_{j,n}$'s such that $P(\norm{X_{j,n}}_2>\epsilon_{j,n})\leq P(\norm{M_{j,n}}_2\geq \eta)$ for all $0\leq j<n-\iota$ under some small enough $\iota$ and $\eta$, then we can achieve the desired upper bound of $\E[n-\bar\tau_n]$. To find $\iota$, $\eta$ and construct $\epsilon_{j,n}$'s that satisfy this condition, we characterize $Z_{j,n}$'s in the following way. Define
\begin{equation*}
    Z_n(j,d) = \E\left[d_{j+1,n}-\delta_{j+1,n}(d_0)\bigg\rvert d_{j,n}=d\right]-(d-\delta_{j,n}(d_0))
\end{equation*}
By the independence property given in Assumption \ref{Assumption: Single-sample} and Assumption, \ref{Assumption: Nonstationarity}
\begin{equation}\label{eqn: ZRelation}
    Z_{j,n}\mathbbm{1}\{j<\tau_n\}=Z_n(j,d_{j,n})\mathbbm{1}\{j<\tau_n\}\quad\mathrm{a.s.}
\end{equation}
In the setting that $\{u_{j,n},a_{j,n}\}_{j=1}^n$ are i.i.d. and satisfy Assumptions \ref{Assumption: Nonstationarity} and \ref{Assumption: p_star}, $\delta_{j,n}(d) = d$ for all $d\in B_{\epsilon_d}(d_0)$, then 
\begin{equation*}
    \begin{aligned}
        Z_{j,n}\mathbbm{1}\{j<\tau_n\} &= \frac{d_{j,n}-\E\left[a_{j+1,n}\mathbbm{1}\left\{u_{j+1,n}>{a_{j+1,n}}^Tp_{j,n}^{\pi}\right\}\bigg\rvert d_{j,n}\right]}{n-j-1}\mathbbm{1}\{j<\tau_n\}\\
        &= \frac{\E\left[a_{j+1,n}\left(\mathbbm{1}\left\{u_{j+1,n}>{a_{j+1,n}}^Tp_{j,n}^*(d_{j,n})\right\}-\mathbbm{1}\left\{u_{j+1,n}>{a_{j+1,n}}^Tp_{j,n}^{\pi}\right\}\right)\bigg\rvert d_{j,n}\right]}{n-j-1}\mathbbm{1}\{j<\tau_n\}
    \end{aligned}
\end{equation*}
Thus, in the i.i.d. case, the uniform dual convergences results imply $Z_{j,n}$'s are small enough so that there is no big difference between $X_{j,n}$ and $M_{j,n}$. Then, we can directly take $\epsilon_{j,n} = \epsilon_d$, which will make our analysis of the stopping time almost the same as the analysis in the literature of the i.i.d. case (ex. \cite{Bray2024}). However, in the nonstationary setting, $\{\delta_{j,n}(d)\}_{j=1}^n$ is no longer a constant sequence, and $X_{j,n}$ and $M_{j,n}$ are different significantly. We need to handle this difference by appropriately constructing $\epsilon_{j,n}$. The following theorem summarizes our approach and provide the main result of this section.
\begin{thm}\label{Theorem: LogStoppingTime}
    Under Assumptions \ref{Assumption: Asymptotic Regime}, \ref{Assumption: Single-sample}, \ref{Assumption: Nonstationarity}, \ref{Assumption: p_star}, there exist constants $\eta$, $C_{Res}$ and $\{\iota_n\}_{n=1}^{+\infty}$ such that for all $P\in\mathcal{P}$, when $n$ is large enough, if $\epsilon_{0,n}=\eta$, and
    \begin{equation*}
        \epsilon_{j+1,n} = \epsilon_{j,n}+\max_{d\in B_{\epsilon_{j,n}}(\delta_{j,n}(d_0))}\norm{Z_n(j,d)}_2\quad 0\leq j<n-\iota_n
    \end{equation*} 
    then
    \begin{equation*}
        \bar\tau_n\geq\tau_n\equiv \inf\left\{1\leq j\leq n-\iota_n-1:d_{j,n}\notin B_{\epsilon_{j,n}}(\delta_{j,n}(d_0))\right\}\wedge n-\iota_n\quad\mathrm{a.s.}
    \end{equation*}
    and
    \begin{equation*}
        \E[n-\bar\tau_n]\leq C_{Res}\log n.    
    \end{equation*}
\end{thm}
A key step in the proof of Theorem \ref{Theorem: LogStoppingTime} is to show that there exists $\eta$ small enough such that $\epsilon_{j,n}<\epsilon_d$ for all $j$ and $n$. This step is based on the fact that Assumption \ref{Assumption: Nonstationarity} guarantees that the distributions of the rewards and the resource consumptions do not change too rapidly over the planning horizon. Proposition \ref{Proposition: Boundedness and Smoothness}(d) gives the precise mathematical description of this fact. Thus, we can view Proposition \ref{Proposition: Boundedness and Smoothness}(d) as a variation budget requirement for achieving single-sample polylogarithmic regret upper bound for nonstationary OLP.
\section{Regret Analysis}\label{Section: Regret Analysis}
Theorem \ref{Theorem: DualConvergence} and Theorem \ref{Theorem: LogStoppingTime} help us deal with most of the challenging parts of the regret analysis of nonstationary OLP. In the following, we use the idea of myopic regret (\cite{Bray2024}) to decompose the regret and then apply the two theorems derived above to prove the polylogarithmic regret upper bound. Define
\begin{equation}\label{eqn: RemainingOptimalOffline}
    R_{j,n}^{off}(d) = \max\left\{\sum_{k=j+1}^nu_{k,n}x_{k,n}:\sum_{k=j+1}^n{a_{k,n}}x_{k,n}\leq (n-j)d,\;x_{k,n}\in [0,1]\;\forall j<k\leq n\right\},
\end{equation}
$R_{j,n}^{off}(d)$ is the optimal total reward earned if we allocate $(n-j)d$ resources to the last $n-j$ orders in an offline fashion. $p_{j,n}(d)$ is then an optimal dual solution corresponding to the resource constraint of (\ref{eqn: RemainingOptimalOffline}). In addition, define 
\begin{equation*}
    R_{j,n}^\pi = \sum_{k=j+1}^nu_{k,n}x_{k,n}^{\pi},   
\end{equation*}
which is the total reward earned by Algorithm \ref{alg: Simulater} on the last $n-j$ orders. Fix $P\in\mathcal{P}$, recall $d_{j,n}=\frac{b_{j,n}}{n-j}$ which is the average remaining resource under Algorithm \ref{alg: Simulater}.
\begin{equation*}
    \begin{aligned}
        &Regret_n(\pi,P)\\
        =& \mathbb{E}\left[R_{0,n}^{off}(d_0)-R_{0,n}^\pi\right]\\
        =& \mathbb{E}\left[\max_{x_{1,n}\in [0,1]\;{a_{1,n}}x_{1,n}\leq nd_0} u_{1,n}x_{1,n} + R_{1,n}^{off}\left(\frac{nd_0-{a_{1,n}}x_{1,n}}{n-1}\right)-u_{1,n}x_{1,n}^{\pi}-R_{1,n}^\pi\right]\\
        =& \mathbb{E}\left[\max_{x_{1,n}\in [0,1]\;{a_{1,n}}x_{1,n}\leq nd_0} u_{1,n}x_{1,n} + R_{1,n}^{off}\left(\frac{nd_0-{a_{1,n}}x_{1,n}}{n-1}\right)-R_{1,n}^{off}(d_{1,n})-u_{1,n}x_{1,n}^{\pi}\right]+\mathbb{E}\left[R_{1,n}^{off}(d_{1,n})-R_{1,n}^\pi\right]
    \end{aligned}
\end{equation*}
By LP Duality Theory, $R_{j,n}^{off}(d) = \min_{p\geq0}\;(n-j)\cdot g_{j,n}(p,d)$. Since $p_{1,n}(d_{1,n})$ is a feasible solution for evaluating $R_{1,n}^{off}\left(\frac{nd_0-{a_{1,n}}x_{1,n}}{n-1}\right)$ for any $x_{1,n}$, 
\begin{equation*}
    \begin{aligned}
        R_{1,n}^{off}\left(\frac{nd_0-{a_{1,n}}x_{1,n}}{n-1}\right)-R_{1:n}^{off}(d_{1,n})&\leq \left(nd_0-{a_{1,n}}x_{1,n}-(n-1)d_{1,n}\right)^Tp_{1,n}(d_{1,n})\\
        &=(x_{1,n}^{\pi}-x_{1,n}){a_{1,n}}^Tp_{1,n}(d_{1,n})
    \end{aligned}
\end{equation*}
Then,
\begin{equation*}
    \begin{aligned}
        Regret_n(\pi,P)\leq &\mathbb{E}\left[\max_{x_{1,n}\in [0,1]} \left(u_{1,n}-{a_{1,n}}^Tp_{1,n}(d_{1,n})\right)\left(x_{1,n}-x_{1,n}^{\pi}\right)\right]+\mathbb{E}\left[R_{1,n}^{off}(d_{1,n})-R_{1,n}(\pi)\right]\\
        =&\mathbb{E}\left[\left(u_{1,n}-{a_{1,n}}^Tp_{1,n}(d_{1,n})\right)\left(\mathbbm{1}\left\{u_{1,n}>{a_{1,n}}^Tp_{1,n}(d_{1,n})\right\}-x_{1,n}^{\pi}\right)\right]+\mathbb{E}\left[R_{1,n}^{off}(d_{1,n})-R_{1,n}^\pi\right]
    \end{aligned}
\end{equation*}
Repeat the above argument for each $j$, we have the following general regret upper bound. Let $\bar\tau_n$ be defined in Section \ref{Section: RemainResource}.
\begin{equation*}
    \begin{aligned}
        Regret_n(\pi,P)\leq&\sum_{j=1}^{n-1}\mathbb{E}\left[\left(u_{j,n}-{a_{j,n}}^Tp_{j,n}(d_{j,n})\right)\left(\mathbbm{1}\left\{u_{j,n}>{a_{j,n}}^Tp_{j,n}(d_{j,n})\right\}-x_{j,n}^{\pi}\right)\right]+\mathbb{E}\left[R_{n-1,n}^{off}(d_{n-1,n})-R_{n-1,n}^\pi\right]\\
        \leq &\sum_{j=1}^{n-1}\E\left[\left(u_{j,n}-{a_{j,n}}^Tp_{j,n}(d_{j,n})\right)\left(\mathbbm{1}\left\{u_{j,n}>{a_{j,n}}^Tp_{j,n}(d_{j,n})\right\}-\mathbbm{1}\{u_{j,n}>{a_{j,n}}^Tp_{j-1,n}^{\pi}\}\right)\mathbbm{1}\{j<\bar\tau_n\}\right]\\
        &+\left(\bar{u}+\frac{m\alpha\bar u}{\underline{d}}\right)\E[n-\tau_n+1] + 2\bar u
    \end{aligned}
\end{equation*}
where the last line is by Lemma \ref{Lemma: BoundedPrice}. The following lemma provides an upper bound on the myopic regret.
\begin{lem}\label{Lemma: MyopicRegret}
    Under Assumptions \ref{Assumption: Asymptotic Regime}, \ref{Assumption: Single-sample}, \ref{Assumption: Nonstationarity}, \ref{Assumption: p_star}, there exists a constant $C_{Myopic}$ and $\tilde\iota$ such that for all $P\in\mathcal{P}$, when $n$ is large enough,
    \begin{equation*}
        \E\left[\left(u_{j,n}-{a_{j,n}}^Tp_{j,n}(d_{j,n})\right)\left(\mathbbm{1}\left\{u_{j,n}>{a_{j,n}}^Tp_{j,n}(d_{j,n})\right\}-\mathbbm{1}\{u_{j,n}>{a_{j,n}}^Tp_{j-1,n}^{\pi}\}\right)\mathbbm{1}\{j<\bar\tau_n\}\right]\leq\frac{C_{Mypoic}\log n}{n-j}
    \end{equation*}
    for all $1\leq j<n-\tilde\iota$.
\end{lem}
Lemma \ref{Lemma: MyopicRegret} is implied by the uniform dual convergence result in Theorem \ref{Theorem: DualConvergence} and the smoothness properties in Proposition \ref{Proposition: Boundedness and Smoothness}. Together with Theorem \ref{Theorem: LogStoppingTime}, we get the polylogarithmic regret upper bound summarized in the following Theorem.
\begin{thm}\label{Theorem: RegretBound}
    Under Assumptions \ref{Assumption: Asymptotic Regime}, \ref{Assumption: Single-sample}, \ref{Assumption: Nonstationarity}, \ref{Assumption: p_star}, there exists a constant $C$ such that for all $P\in\mathcal{P}$, when $n$ is large enough,
    \begin{equation*}
        Regret_n(\pi,P)\leq C (\log n)^2
    \end{equation*}
\end{thm}
\subsection{Future extension: prove $O(\log n)$ regret upper bound}\label{Section: Improve_to_logn}
If we can improve the rate of uniform dual convergence in Theorem \ref{Theorem: DualConvergence} from $O\left(\frac{\log(n-j)}{n-j}\right)$ to $O(\frac{1}{n-j})$, we can improve the regret upper bound from $O((\log n)^2)$ to $O(\log n)$. To improve the uniform dual convergence rate, we only need to show
\begin{equation}\label{eqn: desiredempiricalprocess}
    \mathbb{E}\left[\sup_{p\in\Omega_p}\norm{\frac{1}{n-j}\sum_{k=j+1}^na_{k,n}\mathbbm{1}\{u_{k,n}>{a_{k,n}}^Tp\}-\bar\delta_{k,n}(p)}^2_2\right]=O\left(\frac{1}{n-j}\right)
\end{equation}
\cite{Bray2024} proves (\ref{eqn: desiredempiricalprocess}) in the i.i.d. case by known results in the theory of empirical process. However, when the distributions of the rewards and the resource consumptions are independent but not identical, the theory of empirical process cannot be directly applied to show (\ref{eqn: desiredempiricalprocess}). An immediate next step of this work is to derive conditions under which (\ref{eqn: desiredempiricalprocess}) holds for the nonidentical case, which will then allow us to improve the regret of Algorithm \ref{alg: Simulater} to $O(\log n)$. 
\section{Conclusion}\label{Section: Conclusion}
This paper investigates the nonstationary Online Linear Programming (OLP) problem, where reward-resource consumption pairs are drawn from a sequence of independent but non-identical distributions. We operate within the challenging single-sample setting, assuming the decision-maker has access to only one sample from each distribution at the start of the planning horizon. Based on a dynamic programming perspective, we propose a dual-based re-solving algorithm specifically tailored for this nonstationary environment. Our regret analysis is conducted under the regime where the resource endowment scales linearly with the total number of orders.

A key structural insight from our work is that while the system states evolve nonstationarily under our policy, the dual prices utilized by the algorithm remain relatively stable. To formalize this relationship, we introduce the concept of a $\delta$-path, which serves as a theoretical bridge between the nonstationary evolution of system states and the quasi-stationary behavior of the dual prices. Leveraging the concept of $\delta$-path, we establish uniform dual convergence results and provide a rigorous analysis of the average remaining resource process. We demonstrate that the prices generated by our algorithm closely track the optimal dual solution of the corresponding offline problem. Consequently, we prove that our proposed algorithm achieves an $O(\log^2 n)$ regret bound. Our results demonstrate that polylogarithmic regret is attainable even under significant environmental shifts and minimal data availability, bridging the gap between stationary OLP and more volatile real-world resource allocation problems.

An immediate priority for future research is to refine the $O(\log^2 n)$ regret upper bound to an $O(\log n)$ bound, as discussed in Section \ref{Section: Improve_to_logn}. Broader extensions include generalizing the framework to accommodate dependent distributions, modeling arrivals as a nonstationary Poisson process, and investigating settings where available samples are derived from estimated distributions rather than the true underlying distributions.

\bibliographystyle{plainnat}
\bibliography{reference}

\newpage
\section{Appendix}\label{Section: Appendix}
\subsection{Proof of Example \ref{Example: LinearRegret}}
\begin{proof}
    Denote $b^{\pi}(n)$ be the amount of resource left after making decisions on the order $\left\lceil\frac{n}{2}\right\rceil$. Take $\delta=\frac{1}{128}$ and $\epsilon=\frac{1}{32}$. Define $\mathcal{E}_{n}= \{\frac{(1-\epsilon)n}{4}\leq b^{\pi}(n)\leq \frac{n}{4}\}$, and denote $u_{(j),n}$ as the $j$th largest value among $\{u_{j,n}\}_{j=1}^n$ and $\tilde u_{(j),n}$ as the $j$th largest value among $\{u_{j,n}\}_{j=\left\lceil\frac{n}{2}\right\rceil+1}^n$. Since the distribution of the rewards is completely unknown, $\E_{P_{1,n}}[\mathbbm{1}_{\mathcal{E}_n}]=\E_{P_{2,n}}[\mathbbm{1}_{\mathcal{E}_n}]$. If $\E_{P_{1,n}}[\mathbbm{1}_{\mathcal{E}_n}]\geq 1-\delta$ then on event $\mathcal{E}_{n}$,
    \begin{equation*}
        R_n(\pi)\leq \frac{\epsilon}{4}n + \sum_{j=1}^{\left\lfloor\frac{n}{4}\right\rfloor}\tilde u_{(j),n}
    \end{equation*}
    Then, by the fact that the expectation of the sum of the $k$ largest values among $n$ i.i.d. $Uniform(0,1)$ random variables is $\frac{k(2n-k+1)}{2(n+1)}$,
    \begin{equation*}
        \begin{aligned}
            \E_{P_{1,n}}\left[R_n(\pi)\right] = & \E_{P_{1,n}}\left[R_n(\pi)\mathbbm{1}_{\mathcal{E}_n^c}\right] + \E_{P_{1,n}}\left[R_n(\pi)\mathbbm{1}_{\mathcal{E}_n}\right]\\
            \leq  & \delta n + \frac{\epsilon}{4}n + \E_{P_{1,n}}\left[\sum_{j=1}^{\left\lfloor\frac{n}{4}\right\rfloor}\tilde u_{(j),n}\right]\\
            \leq &\delta n + \frac{\epsilon}{4}n +\frac{\left\lfloor\frac{n}{4}\right\rfloor\left(n - \left\lfloor\frac{n}{4}\right\rfloor + 1\right)}{n+1}
        \end{aligned}
    \end{equation*}
    In addition, 
    \begin{equation*}
        \E_{P_{1,n}}[R_n^*] \geq \E_{P_{1,n}}\left[\sum_{j=1}^{\left\lfloor\frac{n}{4}\right\rfloor} u_{(j),n}\right] = \frac{\left\lfloor\frac{n}{4}\right\rfloor\left(2n - \left\lfloor\frac{n}{4}\right\rfloor + 1\right)}{2\left(n+1\right)}
    \end{equation*}
    Then,
    \begin{equation}\label{eqn: lowerboundexample_1}
        \begin{aligned}
            Regret_n(\pi,P_{1,n})\geq\frac{\left\lfloor\frac{n}{4}\right\rfloor(\left\lfloor\frac{n}{4}\right\rfloor -1 )}{2(n+1)}-\left(\delta+\frac{\epsilon}{4}\right)n\geq \frac{1}{64}n - \frac{3}{8}
        \end{aligned}
    \end{equation}
    If  $\E_{P_{2,n}}[\mathbbm{1}_{\mathcal{E}^c_n}]\geq \delta$, then on event $\mathcal{E}_{n}^c$, 
    \begin{equation*}
        R_n(\pi)\leq \frac{n}{4} - b^{\pi}(n) + \sum_{j=1}^{b^{\pi}(n)} \tilde u_{(j),n}
    \end{equation*}
    Let $\mathcal{F}_{\left\lceil\frac{n}{2}\right\rceil}$ be the sigma algebra generated by $\{u_{j,n}\}_{j=1}^{\left\lceil\frac{n}{2}\right\rceil}$.
    \begin{equation*}
        \begin{aligned}
            \E_{P_{2,n}}\left[R_n(\pi)\mathbbm{1}_{\mathcal{E}_n^c}\right] &\leq \E_{P_{2,n}}\left[\left(\frac{n}{4} - b^{\pi}(n)\right)\mathbbm{1}_{\mathcal{E}_n^c}\right] + \E_{P_{2,n}}\left[\E_{P_{2,n}}\left[\sum_{j=1}^{b^{\pi}(n)} \tilde u_{(j),n}\bigg\rvert\mathcal{F}_{\left\lceil\frac{n}{2}\right\rceil}\right]\mathbbm{1}_{\mathcal{E}_n^c}\right]\\
            &\leq \E_{P_{2,n}}\left[\left(\frac{n}{4} - b^{\pi}(n)\right)\mathbbm{1}_{\mathcal{E}_{n}^c}\right] + \E_{P_{2,n}}\left[\left(\frac{\left\lfloor\frac{n}{4}\right\rfloor\left(2\left\lfloor\frac{n}{2}\right\rfloor - \left\lfloor\frac{n}{4}\right\rfloor + 1\right)}{2\left(\left\lfloor\frac{n}{2}\right\rfloor+1\right)}+2b^{\pi}(n)\right)\mathbbm{1}_{\mathcal{E}_{n}^c}\right]
        \end{aligned}
    \end{equation*}
    In addition, 
    \begin{equation*}
        \begin{aligned}
            \E_{P_{2,n}}[R_n^*\mathbbm{1}_{\mathcal{E}^c}] &\geq \E_{P_{2,n}}\left[\sum_{j=1}^{\left\lfloor\frac{n}{4}\right\rfloor} \tilde u_{(j),n}\mathbbm{1}_{\mathcal{E}_{n}^c}\right]= \mathbb{E}_{P_{2,n}}\left[\left(\frac{\left\lfloor\frac{n}{4}\right\rfloor\left(2\left\lfloor\frac{n}{2}\right\rfloor - \left\lfloor\frac{n}{4}\right\rfloor + 1\right)}{2\left(\left\lfloor\frac{n}{2}\right\rfloor+1\right)}+\left\lfloor\frac{n}{2}\right\rfloor\right)\mathbbm{1}_{\mathcal{E}_{n}^c}\right]
        \end{aligned}
    \end{equation*}
    Thus,
    \begin{equation}\label{eqn: lowerboundexample_2}
        \begin{aligned}
            Regret(\pi,P_{2,n})\geq \mathbbm{E}_{P_{2,n}}\left[\left(\left\lfloor\frac{n}{2}\right\rfloor-\frac{n}{4}-b^{\pi}(n)\right)\mathbbm{1}_{\mathcal{E}_{n}^c}\right]\geq \frac{\epsilon\delta}{4}n - 1
        \end{aligned}
    \end{equation}
    (\ref{eqn: lowerboundexample_1}) and (\ref{eqn: lowerboundexample_2}) together imply that the $Regret_n(\pi)$ is $\Omega(n)$.
\end{proof}
\subsection{Proof of Lemma \ref{Lemma: BoundedPrice-I}}
\begin{proof}
    Since $p_{j,n}^*(d)\geq 0$,
    \begin{equation*}
        \begin{aligned}
            d^{min}e^Tp_{j,n}^*(d)\leq d^Tp^*_{j,n}(d)\leq g^*_{j,n}(p^*_{j,n}(d),d)\leq g^*_{j,n}(0,d)\leq\bar{u} 
        \end{aligned}
    \end{equation*}
\end{proof}
\subsection{Proof of Lemma \ref{Lemma: Riemann Integral}}
\begin{proof}
    Fix $p\in\mathbb{R}^m$, $\frac{1}{n}\sum_{j=1}^n\E\left[\left(u_{j,n}-{a_{j,n}}^Tp\right)^+\right]$ is a Riemann sum of $\int_{0}^{1}\mathbbm{E}\left[(u_s-{a_s}^Tp)^+\right]$. We assume the following that $a_{j,n}$'s have probability density function. For the case that $a_{j,n}$'s have probability mass function. In addition, for any $t_1,t_2\in [0,1]$,
    \begin{equation*}
        \begin{aligned}
            &\abs{\E\left[\left(u_{t_1}-{a_{t_1}}^Tp\right)^+\right]-\E\left[\left(u_{t_2}-{a_{t_2}}^Tp\right)^+\right]}\\
            \leq & \int_{\Xi_a}\int_{0}^{\bar u} (u-a^Tp)^+\abs{f(t_1,u,a)v(t_1,a)-f(t_2,u,a)v(t_2,a)}du\;da\\
            \leq & 2\bar\mu \bar L\abs{t_1-t_2}\int_{\Xi_a}\int_{0}^{\bar u} (u-a^Tp)^+du\;da\\
            \leq & 2\bar u\bar\mu\bar L \alpha^m (\bar u + \alpha p)\abs{t_1-t_2}
        \end{aligned}
    \end{equation*}
    where the Lipschitz continuity used above is given by Assumption \ref{Assumption: Nonstationarity} (c) and (d). Thus, $\mathbbm{E}\left[(u_s-{a_s}^Tp)^+\right]$ is Lipschitz continuous on $[0,1]$ with respect to $t$ and hence is Riemann integrable. Then,
    \begin{equation*}
        \lim_{T\rightarrow+\infty}\frac{1}{T}\sum_{j=1}^T\E\left[\left(u_{j,T}-{a_j}^Tp\right)^+\right]=\int_{0}^{1}\mathbbm{E}\left[(u_s-{a_s}^Tp)^+\right] dt
    \end{equation*}
    which completes the proof.
\end{proof}
\subsection{Proof of Proposition \ref{Proposition: Boundedness and Smoothness}}
\begin{proof} 
    Let $\{e_i\}_{i=1}^m$ be the standard basis of $\mathbb{R}^m$. For part (a), for any $p\in \mathbb{R}^m$, by Assumption \ref{Assumption: Nonstationarity}(d), $\mathbb{P}\left\{u_{j,n}={a_{j,n}}^Tp\right\}=0$, then
    \begin{equation*}
        \lim_{h\rightarrow 0}\frac{(u_{j,n}-{a_{j,n}}^T(p+he_i))^+ - (u_{j,n}-{a_{j,n}}^Tp)^+}{h} = -{e_i}^T a_{j,n}\mathbbm{1}\{u_{j,n}>{a_{j,n}}^Tp\}\quad\mathrm{a.s.}
    \end{equation*}
    and
    \begin{equation*}
        \abs{\frac{(u_{j,n}-{a_{j,n}}^T(p+he_i))^+ - (u_{j,n}-{a_{j,n}}^Tp)^+}{h}}\leq\alpha\quad\forall h
    \end{equation*}
    By Bounded Convergence Theorem, $\E\left[\left(u_{j,n}-{a_{j,n}}^Tp\right)^+\right]$ is differentiable on $\mathbb{R}^m$ and
    \begin{equation*}
        \nabla\E\left[\left(u_{j,n}-{a_{j,n}}^Tp\right)^+\right] = -\E\left[a_{j,n}\mathbbm{1}\{u_{j,n}>{a_{j,n}}^Tp\}\right]
    \end{equation*}    
    By Assumption \ref{Assumption: p_star}, there exists $\epsilon$ small enough such that for any $P\in\mathcal{P}$, $a^Tp\in (0,\bar u)$ if $p\in B_{2\epsilon}(p^*)\equiv\Omega_p\subseteq\mathbb{R}^{+}_{m}$ and $a\in\Xi$. Then, for $p\in\Omega_p$
    \begin{equation*}
        \begin{aligned}
            &\lim_{h\rightarrow0}\frac{\E\left[{e_i}^Ta_{j,n}\mathbbm{1}\{u_{j,n}>{a_{j,n}}^Tp\}-{e_i}^Ta_{j,n}\mathbbm{1}\{u_{j,n}>{a_{j,n}}^T(p+he_k)\}\right]}{h}\\
            =&\lim_{h\rightarrow0}\E\left[\frac{{e_i}^Ta_{j,n}\left(P(u_{j,n}>{a_{j,n}}^Tp|a_{j,n})-P(u_{j,n}>{a_{j,n}}^T(p+he_k))|a_{j,n})\right)}{h}\right]\\
            =&\E\left[({e_i}^Ta_{j,n})({e_k}^Ta_{j,n})f\left(\frac{j}{n},{a_{j,n}}^Tp,a_{j,n}\right)\right]
        \end{aligned}
    \end{equation*}
    where the last line is by Assumption \ref{Assumption: Nonstationarity}(d) and Bounded Convergence Theorem. Then, $\E\left[\left(u_{j,n}-{a_{j,n}}^Tp\right)^+\right]$ is twice differentiable on $\Omega_p$ and its Hessian is
    \begin{equation*}
        \E\left[(a_{j,n}{a_{j,n}}^T)f\left(\frac{j}{n},{a_{j,n}}^Tp,a_{j,n}\right)\right]
    \end{equation*}
    Similarly,
    \begin{equation*}
        \begin{aligned}
            &\lim_{h\rightarrow 0} \frac{\E\left[({e_i}^Ta_{j,n})(e_k^Ta_{j,n})f\left(\frac{j}{n},{a_{j,n}}^T(p+he_q),a_{j,n}\right)\right]-\E\left[({e_i}^Ta_{j,n})({e_k}^Ta_{j,n})f\left(\frac{j}{n},{a_{j,n}}^Tp,a_{j,n}\right)\right]}{h}\\
            =&\E\left[({e_i}^Ta_{j,n})({e_k}^Ta_{j,n})({e_q}^Ta_{j,n})\frac{\partial f}{\partial u}\left(\frac{j}{n},{a_{j,n}}^Tp,a_{j,n}\right)\right]
        \end{aligned}
    \end{equation*}
    Thus, $\E\left[\left(u_{j,n}-{a_{j,n}}^Tp\right)^+\right]$ is three times continuously differentiable on $\Omega_p$. In addition, Assumption \ref{Assumption: Nonstationarity} implies that
    \begin{equation*}
        \abs{\E\left[({e_i}^Ta_{j,n})({e_k}^Ta_{j,n})({e_q}^Ta_{j,n})\frac{\partial f}{\partial u}\left(\frac{j}{n},{a_{j,n}}^Tp,a_{j,n}\right)\right]}\leq\bar L\alpha^3\equiv L
    \end{equation*}
    For part (b), for each $j$ and $n$, the largest eigenvalue of the Hessian of $\E[(u_{j,n}-{a_{j,n}}^Tp)^+]$ is upper bounded by its trace, which is upper bounded by $m\alpha\bar\mu\equiv\bar\sigma$ by Assumption \ref{Assumption: Nonstationarity}. For the smallest eigenvalue of the Hessian of $\E[(u_{j,n}-{a_{j,n}}^Tp)^+]$, note that
    \begin{equation*}
        \begin{aligned}
            \E\left[(a_{j,n}{a_{j,n}}^T)f\left(\frac{j}{n},{a_{j,n}}^Tp,a_{j,n}\right)\right]= \E\left[(a_{j,n}{a_{j,n}}^T)\underline{\mu}\right] + \E\left[(a_{j,n}{a_{j,n}}^T)\left(f\left(\frac{j}{n},{a_{j,n}}^Tp,a_{j,n}\right)-\underline{\mu}\right)\right]
        \end{aligned}
    \end{equation*}
    Both of the two matrices on the right-hand-side are symmetric and positive semi-definite. Then, Assumption \ref{Assumption: Nonstationarity} implies that the smallest eigenvalue of $\E\left[a_{j,n}{a_{j,n}}^T\underline{\mu}\right]$ is greater than $\underline{\mu}\lambda_{min}\equiv\underline{\sigma}>0$. Then, together with the Weyl's Inequality, the smallest eigenvalue of the Hessian of the Hessian of $\E[(u_{j,n}-{a_{j,n}}^Tp)^+]$ is lower bounded by $\underline{\sigma}$. This immediately shows that the eigenvalues of the Hessian of $g^*_{j,n}(p,d_0)$ are in $[\underline{\sigma},\bar\sigma]$. 
    
    For part (c), because $(u_{j,n}-{a_{j,n}}^Tp)^+$ is convex in $p$ almost surely for all $j$and $n$, $g_{n}^*(\cdot,d_0)$ and $g_{\infty}^*(\cdot,d_0)$ are all convex. Together with the pointwise convergence given by Lemma \ref{Lemma: Riemann Integral}, $g_T^*(\cdot,d_0)$ converges to $g_{\infty}^*(\cdot,d_0)$ uniformly over $\left[0,\frac{2\bar u}{d^{min}_0}\right]^m$, where $d^{min}_0$ is defined in Lemma \ref{Lemma: BoundedPrice-I}. Similar arguments in the proof of Lemma \ref{Lemma: BoundedPrice-I} imply that $p^*\in \left[0,\frac{\bar u}{d^{min}_0}\right]^m$ for all $P\in\mathcal{P}$. Together with Assumption \ref{Assumption: p_star}, we can select $\epsilon$ small enough such that $\Omega_p$ is in the interior of $\left[0,\frac{2\bar u}{d^{min}_0}\right]^m$ for all $P\in\mathcal{P}$. Following similar arguments in part (a), the gradient and the Hessian of $g^*_{\infty}(p,d_0)$ on $\Omega_p$ are $d_0 -  \E\left[\int_{0}^1a_s\mathbbm{1}\{u_s>{a_s}^Tp\}ds\right]$ and $\E\left[\int_{0}^1 (a_s{a_s}^T)f(s,{a_s}^Tp, a_s) ds\right]$. In addition, the Hessian is positive definite on $\Omega_p$ with smallest eigenvalue greater than $\underline{\sigma}$. Thus, $p^*$ is the unique optimal solution. Suppose that there exists $\nu>0$ and $\{n_k\}_{k=1}^{+\infty}$ such that $\norm{p_{n_k}^*(d_0)-p^*}_2>\nu$ for all $k$ and all possible selection of $p^*_{n_k}(d_0)$. We can assume $\nu<\epsilon$. Then, the strong convexity of $g^*_{\infty}(\cdot,d_0)$ on $\Omega_p$ implies that, for any $p\in \Omega_p\backslash B_{\nu}(p^*)$,
    \begin{equation*}
        g(p)\geq\frac{\underline{\sigma}}{2}\nu^2
    \end{equation*}
    Together with the fact that the level set of $g_{\infty}^*(\cdot,d_0)$ is convex, we have
    \begin{equation*}
        \{p:g_{\infty}(p,d_0)-g_{\infty}(p^*,d_0)\leq \frac{\underline{\sigma}}{4}\nu^2\}\subseteq B_{\nu}(p^*)
    \end{equation*}
    The uniform convergence implies that when $k$ is large enough,
    \begin{equation*}
        \abs{g_{n_k}^*(p,d_0)-g^*_{\infty}(p,d_0)}<\frac{\underline{\sigma}}{12}\nu^2 \quad\forall p\in \left[0,\frac{2\bar u}{d^{min}_0}\right]^m
    \end{equation*}
    Since $p_{n_k}^*(d_0)\in \left[0,\frac{\bar u}{d^{min}_0}\right]^m$ by Lemma \ref{Lemma: BoundedPrice-I} and $p_{T_k}^*(d_0)\notin B_{\nu}(p^*)$,  when $k$ is large enough
    \begin{equation*}
        g_{T_k}^*(p_{n_k}^*(d_0),d_0)-g_{n_k}^*(p^*,d_0)\geq \frac{\underline{\sigma}}{12}\nu^2
    \end{equation*}
    which means that $p_{n_k}^*$ is not an optimal solution. This is a contradiction. Thus, there exists a selection of $p_n^*(d_0)$'s such that $p_n^*(d_0)\rightarrow p^*$ as $n\rightarrow+\infty$. For this selection, there exists $N_0\in\mathbb{N}_{+}$ such that $\norm{p_N^*(d_0)-p^*}_2<\epsilon$. Then, when $n\geq N_0$, $B_{\epsilon}(p_n^*(d_0))\subseteq\Omega_p$. In addition, $p_n^*(d_0)>0$, and the Hessian of $g^*_n(\cdot,d_0)$ at $p_n^*(d_0)$ is positive definite by part (b), which implies that $p_n^*(d_0)$ is uniquely determined. Furthermore, pick $\kappa$ to be any real number in $(0,1)$, together with Assumption \ref{Assumption: p_star}, we can select $\epsilon$ small enough such that there exists $\tilde \epsilon_p$ which holds for all $P\in\mathcal{P}$ such that $(\inf_{a\in \Xi} a^Tp^*_n(d_0), \sup_{a\in \Xi} a^Tp_n^*(d_0))\subseteq(\tilde\epsilon_{p},\bar u-\tilde\epsilon_p)$. Then, for all $j,n$
    \begin{equation*}
        \begin{aligned}
            \E[a_{j,n}\mathbbm{1}\{u_{j,n}>{a_{j,n}}^Tp_n^*(d_0)\}]\geq \E[a_{j,n}P\{u_{j,n}>{a_{j,n}}^Tp_n^*(d_0)|a_{j,n}\}]\geq \E[a_{j,n}\underline{\mu}\epsilon_p]\geq \underline{a}\underline{\mu}\tilde \epsilon_p e \equiv \epsilon_pe  
        \end{aligned}
    \end{equation*}
    For part (d), By Assumption \ref{Assumption: Nonstationarity}(c) and (d), 
    \begin{equation*}
        \begin{aligned}
            \prod_{j=0}^{n-2} \left(1 + L_{j,n}\right)\leq \prod_{j=0}^{n-2}\left(1+\frac{2\bar\mu\bar L\alpha^{m+2}\sum_{k=j+2}^n\frac{k-(j+1)}{n}}{\underline{\sigma}(n-j)(n-j-1)}\right)\leq \left(1+\frac{\bar\mu\bar L\alpha^{m+2}}{\underline{\sigma}n}\right)^n\leq  e^{\frac{\mu\bar L\alpha^{m+2}}{\underline{\sigma}}}
        \end{aligned}
    \end{equation*}
    
    \begin{equation*}
        \begin{aligned}
            &\prod_{k=0}^{n-2}\left(1+\frac{\alpha^2\sum_{j=k+2}^n\sup_{a\in\Xi,s\in [\underline{s},\bar s]}\abs{f(\frac{j}{n}, s, a)-f(\frac{k+1}{n}, s, a)}}{\underline{\sigma}(n-k)(n-k-1)}\right)\\
            \leq & \prod_{k=0}^{n-2}\left(1+\frac{\alpha^2\bar L\sum_{j=k+2}^n\frac{j-(k+1)}{n}}{\underline{\sigma}(n-k)(n-k-1)}\right)\leq \left(1+\frac{\alpha^2\bar L}{2\underline{\sigma}n}\right)^n\leq  e^{\frac{\alpha^2\bar L}{2\underline{\sigma}}}
        \end{aligned}
    \end{equation*}
\end{proof}
\subsection{Proof of Lemma \ref{Lemma: Unique_Equal_Population_Price}}
\begin{proof}
      Let $N_0$ be given in Proposition \ref{Proposition: Boundedness and Smoothness}(c), and fix $n\geq N_0$. For part (a), fix $d\in \bar{\Omega}_n$. By Proposition \ref{Proposition: Boundedness and Smoothness}(b) and its proof, $g^*_{j,n}(\cdot,d)$ is differentiable and strictly convex on $\Omega_p$ for all $j$, and $\nabla g^*_{j,n}(p,d)=d-\bar{\delta}_{j,n}(p)$. Then, by definition, there exists $p\in\Omega_p$ such that $\nabla g^*_{n}(p,d)=0$. Together with $\Omega_p\subseteq\mathbb{R}^m_{+}$, $\arg\min_{p\geq 0}g^*_{n}(p,d)$ has an unique optimal solution in $\Omega_p$, which proves that $p_n^*(d)\in\Omega_p$ is uniquely determined. Fix $1\leq j< n$, by definition of $\delta_{j,n}(d)$,
      \begin{equation}\label{eqn: delta_and_bar_delta}
          \begin{aligned}
              \delta_{j,n}(d)=\frac{\sum_{k=1}^n\mathbb{E}\left[a_{k,n}\mathbbm{1}\{u_{k,n}>{a_{k,n}}^Tp_n^*(d)\}\right]-\sum_{k=1}^j\mathbb{E}\left[a_{k,n}\mathbbm{1}\{u_{k,n}>{a_{k,n}}^Tp_n^*(d)\}\right]}{n-j}=\bar{\delta}_{j,n}(p_n^*(d))
          \end{aligned}
      \end{equation}
      which implies that $\nabla g^*_{j,n}(p_n^*(d),\delta_{j,n}(d))=0$. Following the same arguments, $p_n^*(d)$ is the unique optimal solution of $\min_{p\geq 0}g^*_{j,n}(p,\delta_{j,n}(d))$. Thus, $p^*_{j,n}(\delta_{j,n}(d))\in\Omega_p$ is uniquely determined, and $p_n^*(d)=p^*_{j,n}(\delta_{j,n}(d))$. 
      This completes the proof of part (a).

      For part (b) and (c), it is sufficient to only prove the case with $j=0$. By Proposition \ref{Proposition: Boundedness and Smoothness} (a) and (b), $\bar{\delta}_n(\cdot)$ is a continuously differentiable function on $\Omega_p$, and its derivative is invertible at each point in $\Omega_p$. Then, by Inverse Function Theorem, $\bar{\delta}_n$ is invertible and an open mapping on $\Omega_p$, and (\ref{eqn: delta_and_bar_delta}) implies that $\bar\delta_n^{-1}(d) = p_n^*(d)$. Proposition \ref{Proposition: Boundedness and Smoothness}(c) implies that $p_n^*(d_0)\in\Omega_p$. Therefore, $d_0\in\bar\Omega_n$. Since the Jacobian matrix of $\bar\delta_n$ on $\Omega_p$ is $-G_{0,n}^*(p)$. Then, part (c) is directly given by Inverse Function Theorem.
\end{proof}
\subsection{Proof of Proposition \ref{Proposition: TailProbDualConv}}
Fix $P\in\mathcal{P}$ satisfying Assumption \ref{Assumption: Nonstationarity} throughout this section. We first introduce some notations. Define
\begin{equation*}
    \nabla g_{j,n}(p,d)=d-\frac{1}{n-j}\sum_{k=j+1}^na_{k,n}\mathbbm{1}\left\{u_{k,n}>{a_{k,n}}^Tp\right\}
\end{equation*}
which is a subgradient of $g_{j,n}(\cdot,d)$ at p, and recall $G^*_{j,n}(p)$ is the Hessian of $g^*_{j,n}(\cdot,d)$ at $p$. By Proposition \ref{Proposition: Boundedness and Smoothness}(b), for any $p\in\Omega_p$, $G^*_{j,n}(p)$ has orthonormal eigenvectors $\left\{\omega_{j,n}^i(p)\right\}_{i=1}^m$ and the corresponding eigenvalues $\left\{\sigma_{j,n}^i(p)\right\}_{i=1}^m$ satisfying $\underline{\sigma}<\sigma_{j,n}^m(p)\leq\cdots\leq\sigma_{j,n}^1(p)<\bar{\sigma}$. 
The following are three useful technical Lemmas.
\begin{lem}\label{Lemma: Convexity}
    ${\nabla g_{j,n}(p,d)}^T(p_{j,n}(d)-p)\leq 0$ for all $d,\;p\in\mathbb{R}_{+}^{m}$ and $j,\;n$.
\end{lem}
\begin{proof}
    Fix $d,\;p\in\mathbb{R}_{+}^{m}$ and $j,\;n$, because $\nabla g_{j,n}(p,d)$ is a subgradient of $g_{j,n}(\cdot,d)$ at $p$, 
    \begin{equation*}
        g_{j,n}(p_{j,n}(d),d)-g_{j,n}(p,d)\geq {\nabla g_{j,n}(p,d)}^T(p_{j,n}(d)-p)
    \end{equation*}
    Since $p_{j,n}(d)$ minimizes $g_{j,n}(\cdot,d)$ over $p\geq 0$, $g_{j,n}(p_{j,n}(d),d)-g_{j,n}(p,d)\leq 0$, which completes the proof.
\end{proof}
\begin{lem}\label{Lemma: LocateOfflinePrice}
    When $n$ is large enough, for all $j$ and $d\in\bar{\Omega}_{n}$, if $p\in\mathbb{R}^m_{+}$ and $\eta> 0$ satisfy $p+k\eta\omega_{j,n}^i(p_n^*(d))\in\mathbb{R}^m_{+}$ and
    \begin{equation*}
        \norm{\nabla g_{j,n}(p+k\eta\omega_{j,n}^i(p_n^*(d)),\delta_{j,n}(d))-k\eta\sigma_{j,n}^i(p_n^*(d))\omega_{j,n}^i(p_n^*(d))}_2\leq\frac{\eta\sigma_{j,n}^m(p_n^*(d))}{2\sqrt{m}}
    \end{equation*}
    for all $1\leq i\leq m$ and $k\in \{-1,1\}$, then 
    \begin{equation*}
        \norm{p_{j,n}(\delta_{j,n}(d))-p}_2\leq \eta\left(1+\frac{2\sqrt{m}\sigma_{j,n}^1(p_n^*(d))}{\sigma_{j,n}^m(p_n^*(d))}\right)
    \end{equation*}
\end{lem}
\begin{proof}
    Let $N_0$ be given by Proposition \ref{Proposition: Boundedness and Smoothness}. Fix $n\geq N_0$ and $d\in\bar\Omega_n$. By Lemma \ref{Lemma: Unique_Equal_Population_Price}, $p_n^*(d)\in\Omega_p$, which makes $\left\{\omega_{j,n}^i(p_n^*(d))\right\}_{i=1}^m$ and $\left\{\sigma_{j,n}^i(p_n^*(d))\right\}_{i=1}^m$ well defined. Let $p$ satisfy the required conditions, by Lemma \ref{Lemma: Convexity},
    \begin{equation*}
        \begin{aligned}
            0\geq &\left(p_{j,n}(\delta_{j,n}(d))-p-k\eta\omega_{j,n}^i(p_n^*(d))\right)^T\nabla g_{j,n}(p+k\eta\omega_{j,n}^i(p_n^*(d)),\delta_{j,n}(d))\\
            =&\left(p_{j,n}(\delta_{j,n}(d))-p-k\eta\omega_{j,n}^i(p_n^*(d))\right)^Tk\eta\sigma_{j,n}^i(p_n^*(d))\omega_{j,n}^i(p_n^*(d))\\
            &+\left(p_{j,n}(\delta_{j,n}(d))-p-k\eta\omega_{j,n}^i(p_n^*(d))\right)^T\left(\nabla g_{j,n}(p+k\eta\omega_{j,n}^i(p_n^*(d)),\delta_{j,n}(d))-k\eta\sigma_{j,n}^i(p_n^*(d))\omega_{j,n}^i(p_n^*(d))\right)\\
            \geq &k\eta\sigma_{j,n}^i(p_n^*(d))\left(p_{j,n}(\delta_{j,n}(d))-p\right)^T\omega_{j,n}^i(p_n^*(d))-\eta^2\sigma_{j,n}^i(p_n^*(d)){\omega_{j,n}^i(p_n^*(d))}^T\omega_{j,n}^i(p_n^*(d))\\
            &-\norm{p_{j,n}(\delta_{j,n}(d))-p-k\eta\omega_{j,n}^i(p_n^*(d))}_2\norm{\nabla g_{j,n}(p+k\eta\omega_{j,n}^i(p_n^*(d)),\delta_{j,n}(d))-k\eta\sigma_{j,n}^i(p_n^*(d))\omega_{j,n}^i(p_n^*(d))}_2
        \end{aligned}
    \end{equation*}
    where the last line is by Cauchy Inequality. Since $\omega_{j,n}^i(p_n^*(d))$'s are orthonormal and the above inequality is true for both $k=1$ and $k=-1$, 
    \begin{equation*}
        \begin{aligned}
            0\geq &\eta\sigma_{j,n}^m(p_n^*(d))\abs{\left(p_{j,n}(\delta_{j,n}(d))-p\right)^T\omega_{j,n}^i(p_n^*(d))}-\eta^2\sigma_{j,n}^1(p_n^*(d))\\
            &-\left(\norm{p_{j,n}(\delta_{j,n}(d))-p}_2+\eta\right)\frac{\eta\sigma_{j,n}^m(p_n^*(d))}{2\sqrt{m}}
        \end{aligned} 
    \end{equation*}
    $\omega_{j,n}^i(p_n^*(d))$'s are orthonormal also implies that there exists $i$ such that
    \begin{equation*}
        \abs{\left(p_{j,n}(\delta_{j,n}(d))-p\right)^T\omega_{j,n}^i(p_n^*(d))}\geq\frac{\norm{p_{j,n}(\delta_{j,n}(d))-p}_2}{\sqrt{m}}
    \end{equation*}
    Then, 
    \begin{equation*}
        \begin{aligned}
            0\geq &\frac{\eta\sigma_{j,n}^m(p_n^*(d))\norm{p_{j,n}(\delta_{j,n}(d))-p}_2}{\sqrt{m}}-\eta^2\sigma_{j,n}^1(p_n^*(d))\\
            &-\left(\norm{p_{j,n}(\delta_{j,n}(d))-p}_2+\eta\right)\frac{\eta\sigma_{j,n}^m(p_n^*(d))}{2\sqrt{m}}
        \end{aligned}
    \end{equation*}
    which is equivalent to 
    \begin{equation*}
        \norm{p_{j,n}(\delta_{j,n}(d))-p}_2\leq \eta\left(1+\frac{2\sqrt{m}\sigma_{j,n}^1(p_n^*(d))}{\sigma_{j,n}^m(p_n^*(d))}\right)
    \end{equation*}
\end{proof}
\begin{lem}\label{Lemma: EmpiricalProcess}
    Under Assumption \ref{Assumption: Asymptotic Regime}, \ref{Assumption: Nonstationarity} and Assumption \ref{Assumption: p_star}, for any $P\in\mathcal{P}$,
    \begin{enumerate}[label=(\alph*)]
        \item \begin{equation*}
        \mathbb{P}\left\{\sup_{p\in\Omega_p}\norm{\frac{1}{n-j}\sum_{k=j+1}^na_{k,n}\mathbbm{1}\{u_{k,n}>{a_{k,n}}^Tp\}-\bar\delta_{k,n}(p)}_2>\nu\right\}\leq 
        \left(1+\left\lceil\frac{6m\bar u\bar\sigma}{d^{min}_0\nu}\right\rceil\right)^mm\exp\left(-\frac{2(n-j)\nu^2}{9m\alpha^2}\right)
    \end{equation*}
    \item When $n-j$ is large enough,
    \begin{equation*}
        \mathbb{E}\left[\sup_{p\in\Omega_p}\norm{\frac{1}{n-j}\sum_{k=j+1}^Ta_{k,n}\mathbbm{1}\{u_{k,n}>{a_{k,n}}^Tp\}-\bar\delta_{k,n}(p)}^2_2\right]\leq m\left(2^{m-1}+\frac{1}{2}\left(\frac{18m\bar u\bar\sigma}{d^0_{min}}\right)^m\right) \frac{9m^2\alpha^2\log(n-j)}{2(n-j)}
    \end{equation*}
    \end{enumerate}
\end{lem}
\begin{proof}
    For part (a), by Hoeffding Inequality, for any $p\in \mathbb{R}^m$,
    \begin{equation*}
        \begin{aligned}
            \mathbb{P}\left\{\norm{\frac{1}{n-j}\sum_{k=j+1}^na_{k,n}\mathbbm{1}\{u_{k,n}>{a_{k,n}}^Tp\}-\bar\delta_{k,n}(p)}_2>\frac{\nu}{3\sqrt{m}}\right\}\leq  m\exp\left(-\frac{2(n-j)\nu^2}{9m\alpha^2}\right)
        \end{aligned}
    \end{equation*}
    We can assume $\Omega_p$ is small enough such that it is a subset of a hyper cube $\bar\Omega_p$, and Proposition \ref{Proposition: Boundedness and Smoothness} holds if we replace $\Omega_p$ in its statements by $\bar\Omega_p$. We can also assume $\bar\Omega_p\subseteq \left[0,\frac{2\bar u}{d^{min}_0}\right]^m$ where $d_0^{min}$ is the smallest entry of $d_0$. We can denote $\bar\Omega_p$ as $\otimes_{i=1}^m \left[\underline{p_i},\bar p_i\right]$ Then, we construct the following grid $\mathcal{G}$
    \begin{equation*}
        \mathcal{G} = \left\{\left(\left(\underline{p_1}+\frac{\nu}{3m\bar\sigma}k_1\right)\wedge\bar p_1,\cdots,\left(\underline{p_m}+\frac{\nu}{3m\bar\sigma}k_m\right)\wedge\bar p_m\right):0\leq k_i\leq \left\lceil\frac{3m(\bar p_i-\underline{p_i})\bar\sigma}{\nu}\right\rceil,\;1\leq i\leq m\right\}
    \end{equation*}
    Then, for any $p\in\Omega_p$, there exists $p_1,p_2\in\mathcal{G}$ such that $p_1\leq p\leq p_2$ and $\norm{p_1-p_2}_2\leq \frac{\nu}{3\sqrt{m}\bar\sigma}$. ``$\leq$'' between vectors element-wise smaller than or equal to. By Mean Value Theorem, $\norm{\bar\delta_{j,n}(p_1)-\bar\delta_{j,n}(p_2)}_2\leq\frac{\nu}{3\sqrt{m}}$. By monotonicity, $\bar\delta_{j,n}(p_2)\leq\bar\delta_{j,n}(p)\leq\bar\delta_{j,n}(p_1)$ and
    \begin{equation*}
        \begin{aligned}
            &\frac{1}{n-j}\sum_{k=j+1}^na_{k,n}\mathbbm{1}\{u_{k,n}>{a_{k,n}}^Tp_2\}\leq \frac{1}{n-j}\sum_{k=j+1}^na_{k,n}\mathbbm{1}\{u_{k,n}>{a_{k,n}}^Tp\}\leq \frac{1}{n-j}\sum_{k=j+1}^na_{k,n}\mathbbm{1}\{u_{k,n}>{a_{k,n}}^Tp_1\}\quad\mathrm{a.s.}
        \end{aligned}
    \end{equation*}
    Define
    \begin{equation*}
        \mathcal{E} = \left\{\norm{\frac{1}{n-j}\sum_{k=j+1}^Ta_{k,n}\mathbbm{1}\{u_{k,n}>{a_{k,n}}^Tp\}-\bar\delta_{j,n}(p)}_2\leq \frac{\nu}{3\sqrt{m}}:p\in\mathcal{G}\right\}
    \end{equation*}
    Let $\{e_i\}_{i=1}^m$ be the standard basis, then on $\mathcal{E}$, for any $p\in\Omega_p$,
    \begin{equation*}
        {e_i}^T\bar\delta_{j,n}(p),{e_i}^T\left(\frac{1}{n-j}\sum_{k=j+1}^na_{k,n}\mathbbm{1}\{u_{k,n}>{a_{k,n}}^Tp\}\right)\in \left({e_i}^T\bar\delta_{j,n}(p_2)-\frac{\nu}{3\sqrt{m}},{e_i}^T\bar\delta_{j,n}(p_1)+\frac{\nu}{3\sqrt{m}}\right)
    \end{equation*}
    which implies that
    \begin{equation*}
        \norm{\frac{1}{n-j}\sum_{k=j+1}^na_{k,n}\mathbbm{1}\{u_{k,n}>{a_{k,n}}^Tp\}-\bar\delta_{j,n}(p)}_2\leq\nu\quad\forall p\in\Omega_p
    \end{equation*}
    Thus,
    \begin{equation*}
        \begin{aligned}
            &\mathbb{P}\left\{\sup_{p\in\Omega_p}\norm{\frac{1}{n-j}\sum_{k=j+1}^na_{k,n}\mathbbm{1}\{u_{k,n}>{a_{k,n}}^Tp\}-\bar\delta_{j,n}(p)}_2>\nu\right\}\\
            \leq & 1 - \mathbb{P}\{\mathcal{E}\}\leq \abs{\mathcal{G}}m\exp\left(-\frac{2(n-j)\nu^2}{9m\alpha^2}\right)\leq \left(1+\left\lceil\frac{6m\bar u\bar\sigma}{d^0_{min}\nu}\right\rceil\right)^mm\exp\left(-\frac{2(n-j)\nu^2}{9m\alpha^2}\right)
        \end{aligned}
    \end{equation*}
    For part (b),
    \begin{equation*}
        \begin{aligned}
            &\mathbb{E}\left[\sup_{p\in\Omega_p}\norm{\frac{1}{n-j}\sum_{k=j+1}^na_{k,n}\mathbbm{1}\{u_{k,n}>{a_{k,n}}^Tp\}-\bar\delta_{k,n}(p)}^2_2\right]\\
            \leq &\frac{1}{2}\int_{0}^{+\infty}1\wedge\left\{\left(1+\left\lceil\frac{6m\bar u\bar\sigma}{d^0_{min}\sqrt{\nu}}\right\rceil\right)^mm\exp\left(-\frac{2(n-j)\nu}{9m\alpha^2}\right)\right\}dv\\
            \leq & 2^{m-1}m\int_0^{+\infty}\exp\left(-\frac{2(n-j)\nu}{9m\alpha^2}\right)dv+\frac{1}{2}\left(\frac{18m\bar u\bar\sigma}{d^0_{min}}\right)^mm\int_{0}^{+\infty}1\wedge \left\{v^{-\frac{m}{2}}\exp\left(-\frac{2(n-j)\nu}{9m\alpha^2}\right)\right\}dv
        \end{aligned}
    \end{equation*}
    where the second inequality is by the fact that $1+\lceil x\rceil \leq 3x$ if $x\geq 1$. To evaluate the right-hand-side. First,
    \begin{equation*}
        \int_0^{+\infty}\exp\left(-\frac{2(n-j)\nu}{9m\alpha^2}\right)dv=\frac{9m\alpha^2}{2(n-j)}
    \end{equation*}
    Secondly, let $v=\frac{9m^2\alpha^2\log(n-j)}{4(n-j)}v'$
    \begin{equation*}
        \begin{aligned}
            &\int_{0}^{+\infty}1\wedge v^{-\frac{m}{2}}\exp\left(-\frac{2(n-j)\nu}{9m\alpha^2}\right)dv\\
            =&\int_{0}^{+\infty}1\wedge \exp\left(-\frac{2(n-j)\nu}{9m\alpha^2}+\frac{m}{2}\log\left(\frac{1}{v}\right)\right)dv\\
            =&\frac{9m^2\alpha^2\log(n-j)}{4(n-j)}\int_{0}^{+\infty}1\wedge\exp\left(-\frac{m}{2}\log(n-j)v'+\frac{m}{2}\log\left(\frac{4(n-j)}{9m^2\alpha^2\log(n-j)\alpha^2v'}\right)\right)dv'
        \end{aligned}
    \end{equation*}
    When $n-j$ is large enough,
    \begin{equation*}
        \begin{aligned}
            &\int_{0}^{+\infty}1\wedge\exp\left(-\frac{m}{2}\log(n-j)v'+\frac{m}{2}\log\left(\frac{4(n-j)}{9m^2\alpha^2\log(n-j)\alpha^2v'}\right)\right)dv'\\
            \leq & 1 + \int_{1}^{+\infty}\exp\left(-\frac{m}{2}\log(n-j)v'+\frac{m}{2}\log\left(\frac{n-j}{v'}\right)\right)dv'\\
            \leq & 1 + \int_{1}^{+\infty}\exp\left(-\frac{m}{2}\log(n-j)v'+\frac{m}{2}\log\left(n-j\right)\right)dv'\\
            \leq &1 + \frac{2}{m\log (n-j)}\leq 2
        \end{aligned}
    \end{equation*}
    Summarize all the above, when $n-j$ is large enough,
    \begin{equation*}
        \begin{aligned}
            &\mathbb{E}\left[\sup_{p\in\Omega_p}\norm{\frac{1}{n-j}\sum_{k=j+1}^na_{k,n}\mathbbm{1}\{u_{k,n}>{a_{k,n}}^Tp\}-\bar\delta_{j,n}(p)}^2_2\right]\leq m\left(2^{m-1}+\frac{1}{2}\left(\frac{18m\bar u\bar\sigma}{d^0_{min}}\right)^m\right) \frac{9m^2\alpha^2\log(n-j)}{2(n-j)}
        \end{aligned}
    \end{equation*}
\end{proof}
\noindent We can now give the complete proof of Proposition \ref{Proposition: TailProbDualConv}. Fix $P\in\mathcal{P}$ and let $N_0$ be given by Proposition \ref{Proposition: Boundedness and Smoothness}(c). Fix $n\geq N_0$, by the same arguments in the proof of Lemma \ref{Lemma: Unique_Equal_Population_Price}, $p_{j,n}^*(d)\in B_{\frac{\epsilon}{2}}(p_n^*(d_0))$ for all $d\in\Omega_{j,n}$. Together with Proposition \ref{Proposition: Boundedness and Smoothness}(c), there exists $\nu<\frac{\epsilon}{2}$ that holds for all $P\in\mathcal{P}$ such that,
    \begin{equation*}
        p_{j,n}^*(d)+\frac{\nu k\omega^i_{j,n}(p_{j,n}^*(d))}{1+\frac{8\sqrt{m}\bar\sigma}{\underline{\sigma}}}\in\Omega_p\quad\forall d\in\Omega_{j,n},\;1\leq i\leq m,\;0\leq j<n
    \end{equation*}
    Denote $\eta(\nu) = \frac{\nu}{1+\frac{8\sqrt{m}\bar\sigma}{\underline{\sigma}}}$ By Lemma \ref{Lemma: Unique_Equal_Population_Price} and Lemma \ref{Lemma: LocateOfflinePrice}, for all $0\leq j < n$,
    \begin{equation}\label{Equation: TailProbDualConv_1}
        \begin{aligned}
            &\mathbb{P}\left\{\sup_{d\in\Omega_{j,n}}\norm{p_{j,n}(d)-p_{j,n}^*(d)}_2>\nu\right\}\\
            \leq & \mathbb{P}\left\{\sup_{d\in\Omega_{j,n}}\norm{p_{j,n}(d)-p_{j,n}^*(d)}_2>\left(1+\frac{2\sqrt{m}\bar\sigma}{\underline{\sigma}}\right)\eta(\nu)\right\}\\
            \leq & \mathbb{P}\left\{\sup_{d\in\Omega_{j,n}}\max_{1\leq i\leq m}\max_{k\in \{-1,1\}}\norm{\nabla g_{j,n}(p_{j,n}^*(d)+k\eta(\nu)\omega^i_{j,n}(p^*_{j,n}(d));d))-k\eta(\nu)\sigma^i_{j,n}(p_{j,n}^*(d))\omega^i_{j,n}(p_{j,n}^*(d))}_2>\frac{\underline{\sigma}\eta(\nu)}{4\sqrt{m}}\right\}
        \end{aligned}
    \end{equation} 
    By Proposition \ref{Proposition: Boundedness and Smoothness}(b),
    \begin{equation*}
        \begin{aligned}
            &\norm{\nabla g^*_{j,n}(p_{j,n}^*(d)+k\eta(\nu)\omega^i_{j,n}(p^*_{j,n}(d));d))-k\eta(\nu)\sigma^i_{j,n}(p_{j,n}^*(d))\omega^i_{j,n}(p_{j,n}^*(d))}_2\\
            =&\norm{\nabla g^*_{j,n}(p_{j,n}^*(d)+k\eta(\nu)\omega^i_{j,n}(p^*_{j,n}(d));d))-\nabla g^*_{j,n}(p_{j,n}^*(d);d))-k\eta(\nu)\sigma^i_{j,n}(p_{j,n}^*(d))\omega^i_{j,n}(p_{j,n}^*(d))}_2\\
            =&\norm{G^*_{j,n}(p_{j,n}^*(d))(k\eta(\nu)\omega^i_{j,n}(p^*_{j,n}(d)))+\mathcal{E}_{j,n}(\nu;i,k,d) -k\eta(\nu)\sigma^i_{j,n}(p_{j,n}^*(d))\omega^i_{j,n}(p_{j,n}^*(d))}_2\\
            =&\norm{\mathcal{E}_{j,n}(\nu;i,k,d)}_2
        \end{aligned}
    \end{equation*}
    where $\mathcal{E}(\cdot)$ denotes the residual term in Taylor's Theorem. Proposition \ref{Proposition: Boundedness and Smoothness}(a) guarantees that $\mathcal{E}(\cdot)$ is $o(\nu)$ uniformly over $j$, $n\geq N_0$, $i$, $k$, and $d\in\Omega_{j,n}$. In addition, we can assume that $\nu$ are small enough such that, for all $0\leq j<n$ 
    \begin{equation}\label{Equation: TailProbDualConv_2}
        \begin{aligned}
            \sup_{d\in\Omega_{j,n}}\max_{1\leq i\leq m}\max_{k\in \{-1,1\}}\norm{\nabla g^*_{j,n}(p_{j,n}^*(d)+k\eta(\nu)\omega^i_{j,n}(p^*_{j,n}(d));d))-k\eta(\nu)\sigma^i_{j,n}(p_{j,n}^*(d))\omega^i_{j,n}(p_{j,n}^*(d))}_2
        \end{aligned}\leq \frac{\eta(\nu)\underline{\sigma}}{8\sqrt{m}}
    \end{equation}
    Then, by (\ref{Equation: TailProbDualConv_1}), (\ref{Equation: TailProbDualConv_2}), and Lemma \ref{Lemma: EmpiricalProcess}, there exists $C_1$ and $C_2$ that holds for all $P\in\mathcal{P}$ such that
    \begin{equation*}
        \begin{aligned}
            &\mathbb{P}\left\{\sup_{d\in\Omega_{j,n}(\nu)}\norm{p_{j,n}(d)-p_{j,n}^*(d)}_2>\nu\right\}\\
            \leq & \mathbb{P}\left\{\sup_{d\in\Omega_{j,n}(\nu)}\max_{1\leq i\leq m}\max_{k\in \{-1,1\}}\norm{\nabla g_{j,n}(p_{j,n}^*(d)+k\eta(\nu)\omega^i_{j,n}(p^*_{j,n}(d));d))-\nabla g^*_{j,n}(p_{j,n}^*(d)+k\eta(\nu)\omega^i_{j,n}(p^*_{j,n}(d));d))}_2>\frac{\underline{\sigma}\eta(\nu)}{8\sqrt{m}}\right\}\\
            \leq & \mathbb{P}\left\{\sup_{p\in\Omega_p}\norm{\frac{1}{n-j}\sum_{k=j+1}^Ta_{k,n}\mathbbm{1}\{u_{k,n}>{a_{k,n}}^Tp\}-\bar\delta_{k,n}(p)}_2>\frac{\underline{\sigma}\eta(\nu)}{8\sqrt{m}}\right\}\leq  C_1 \exp\left(-C_2\nu^2(n-j)\right)
        \end{aligned}
    \end{equation*}
\subsection{Proof of Lemma \ref{Lemma: BoundedPrice}}
    Fix $P\in\mathcal{P}$. Let $N_0$ be given by Proposition \ref{Proposition: Boundedness and Smoothness}(c) and fix $n\geq N_0$. By definition and Assumption \ref{Assumption: Nonstationarity}, 
    \begin{equation*}
        \norm{\bar\delta_{j,n}(p)}_{\infty}\leq \frac{1}{n-j}\sum_{k=j+1}^n\norm{\E[a_{k,n}]}_{\infty}\leq\alpha
    \end{equation*}
    Proposition \ref{Proposition: Boundedness and Smoothness}(c) and Lemma \ref{Lemma: Unique_Equal_Population_Price}(a) together show that $\delta_{j,n}(d_0)>\frac{1-\kappa}{2}\epsilon_p\cdot e$ for all $j$. By Proposition \ref{Proposition: Boundedness and Smoothness}(b) and Mean Value Theorem, for any $p\in B_{\frac{\epsilon}{2}}(p_n^*(d_0))$, 
    \begin{equation*}
        \norm{\bar{\delta}_{j,n}(p)-\bar{\delta}_{j,n}(p_n^*(d_0))}\leq\frac{\bar{\sigma}\epsilon}{2}
    \end{equation*}
    We can assume $\epsilon$ is small enough such that $\frac{\bar\sigma\epsilon}{2} < \frac{1-\kappa}{2}\epsilon_p$. Thus, there exists $\underline{d}>0$ such that each entry of $d$ is greater than $\underline{d}$ when $d\in\Omega_{j,n}$ for all $j$. Then, by the same argument in the proof of Lemma \ref{Lemma: BoundedPrice-I}, we can show that $\norm{p_{j,n}(d)}_2\leq \norm{p_{j,n}(d)}_1\leq \frac{\bar u}{\underline{d}}$ almost surely for all $d\in\Omega_{j,n}$ and $j$.
\subsection{Proof of Theorem \ref{Theorem: DualConvergence}}
Fix $P\in\mathcal{P}$. Let $N_0$ be given by Proposition \ref{Proposition: Boundedness and Smoothness}(c) and fix $n\geq N_0$. Then, by Proposition \ref{Proposition: Boundedness and Smoothness}(a), (b) and Taylor's theorem, we can assume $\epsilon$ is small (so that $\Omega_p$ is small enough) enough such that for any $p\in\Omega_p$, $d\in\Omega_{j,n}$ and $0\leq j<n$,
    \begin{equation}\label{Equation: DC_1}
        (p-p_{j,n}^*(d))^TG^*_{j,n}(p_{j,n}^*(d))(p-p_{j,n}^*(d))\geq\underline{\sigma}\norm{p-p_{j,n}^*(d)}_2^2
    \end{equation}
    and
    \begin{equation}\label{Equation: DC_2}
        \begin{aligned}
            &\norm{\nabla g^*_{j,n}(p;d)-G^*_{j,n}(p_{j,n}^*(d))(p-p_{j,n}^*(d))}_2\\
            =&\norm{\nabla g^*_{j,n}(p;d)-\nabla g^*_{j,n}(p_{j,n}^*(d);d)-G^*_{j,n}(p_{j,n}^*(d))(p-p_{j,n}^*(d))}_2\\
            =&\norm{G^*_{j,n}(p_{j,n}^*(d))(p-p_{j,n}^*(d))+\mathcal{E}(p-p_{j,n}^*(d);j,n,d)-G^*_{j,n}(p_{j,n}^*(d))(p-p_{j,n}^*(d))}_2\\
            =&\norm{\mathcal{E}(p-p_{j,n}^*(d);j,n,d)}_2\\
            \leq &\frac{\underline{\sigma}\norm{p-p_{j,n}^*(d)}}{2}
        \end{aligned}
    \end{equation}
    (\ref{Equation: DC_1}) and (\ref{Equation: DC_2}) then imply that
    \begin{equation}\label{Equation: DC_3}
        \begin{aligned}
            &(p-p_{j,n}^*(d))^T\nabla g^*_{j,n}(p;d)\\
            =&(p-p_{j,n}^*(d))^TG^*_{j,n}(p_{j,n}^*(d))(p-p_{j,n}^*(d))+(p-p_{j,n}^*(d))^T(\nabla g^*_{j,n}(p;d)-G^*_{j,n}(p_{j,n}^*(d))(p-p_{j,n}^*(d)))\\
            \geq & \underline{\sigma}\norm{p-p_{j,n}^*(d)}_2^2 - \norm{\nabla g^*_{j,n}(p;d)-G^*_{j,n}(p_{j,n}^*(d))(p-p_{j,n}^*(d))}_2\norm{p-p_{j,n}^*(d)}_2\\
            \geq &\frac{\underline{\sigma}\norm{p-p_{j,n}^*(d)}_2^2}{2}
        \end{aligned}
    \end{equation}
    Let $\hat{p}_{j,n}(d)=\frac{p_{j,n}(d)+p_{j,n}^*(d)}{2}$. By Lemma \ref{Lemma: Convexity},
    \begin{equation}\label{Equation: DC_4}
        (\hat{p}_{j,n}(d)-p_{j,n}^*(d))^T\nabla g_{j,n}(\hat p_{j,n}(d),d) = (p_{j,n}(d)-\hat{p}_{j,n}(d))^T\nabla g_{j,n}(\hat{p}_{j,n}(d),d)\leq 0 
    \end{equation}
    Fix $0\leq j < n$ and $d\in\Omega_{j,n}$. Then, if $p_{j,n}(d)\in \Omega_P$, $\hat{p}_{j,n}(d)\in \Omega_p$. Together with (\ref{Equation: DC_3}) and (\ref{Equation: DC_4}), if $p_{j,n}(d)\in \Omega_P$, then
    \begin{equation*}
        \begin{aligned}
            &\norm{\hat{p}_{j,n}(d)-p_{j,n}^*(d)}_2\norm{\nabla g^*_{j,n}(\hat{p}_{j,n}(d);d)-\nabla g_{j,n}(\hat{p}_{j,n}(d);d)}_2\\
            \geq & (\hat{p}_{j,n}(d)-p_{j,n}^*(d))^T(\nabla g^*_{j,n}(\hat{p}_{j,n}(d);d)-\nabla g_{j,n}(\hat{p}_{j,n}(d);d))\\
            \geq & (\hat{p}_{j,n}(d)-p_{j,n}^*(d))^T\nabla g^*_{j,n}(\hat{p}_{j,n}(d);d)\\
            \geq &\frac{\underline{\sigma}\norm{\hat{p}_{j,n}(d)-p_{j,n}^*(d)}_2^2}{2}
        \end{aligned}
    \end{equation*}
    which implies that
    \begin{equation}\label{equation: DC_4_1}
        \norm{\nabla g^*_{j,n}(\hat{p}_{j,n}(d);d)-\nabla g_{j,n}(\hat{p}_{j,n}(d);d)}_2\geq\frac{\underline{\sigma}\norm{\hat{p}_{j,n}(d)-p_{j,n}^*(d)}_2}{2}=\frac{\underline{\sigma}\norm{p_{j,n}(d)-p_{j,n}^*(d)}_2}{4}
    \end{equation}
    Let $\nu$ be given by Proposition \ref{Proposition: TailProbDualConv}. By the definition of $\Omega_{j,n}$, $p_{j,n}^*(d)\in B_{\frac{\epsilon}{2}}(p_n^*(d_0))$. Together with the fact $\nu<\frac{\epsilon}{2}$ and Proposition \ref{Proposition: Boundedness and Smoothness}(c), if $p_{j,n}(d)\in B_{\nu}(p_{j,n}^*(d))$, then $p_{j,n}(d)\in\Omega_p$. Thus, (\ref{equation: DC_4_1}) and Lemma \ref{Lemma: EmpiricalProcess}(b) imply that, when $n-j$ is greater than some constant $N_1$ that holds for all $P\in\mathcal{P}$,
    \begin{equation}\label{Equation: DC_5}
        \begin{aligned}
            &\E\left[\sup_{d\in\Omega_{j,n}}\mathbbm{1}\{p_{j,n}(d)\in B_{\nu}(p_{j,n}^*(d))\}\norm{p_{j,n}(d)-p_{j,n}^*(d)}_2^2\right]\\
            \leq &\frac{16}{\underline{\sigma}^2}\E\left[\sup_{d\in\Omega_{j,n}}\norm{\nabla g^*_{j,n}(\hat{p}_{j,n}(d),d)-\nabla g_{j,n}(\hat{p}_{j,n}(d),d)}_2^2\right]\\
            \leq &\frac{16}{\underline{\sigma}^2}\mathbb{E}\left[\sup_{p\in\Omega_p}\norm{\frac{1}{n-j}\sum_{k=j+1}^na_{k,n}\mathbbm{1}\{u_{k,n}>{a_{k,n}}^Tp\}-\bar\delta_{k,n}(p)}_2^2\right]\leq \frac{C_3\log (n-j)}{n-j}
        \end{aligned}
    \end{equation}
    where $C_3$ is a constant that holds for all $P\in\mathcal{P}$ and does not depend on $j$ and $n$. By Proposition \ref{Proposition: TailProbDualConv} and Lemma \ref{Lemma: BoundedPrice}, 
    \begin{equation}\label{Equation: DC_6}
        \begin{aligned}
            &\E\left[\sup_{d\in\Omega_{j,n}}\mathbbm{1}\{p_{j,n}(d)\notin B_{\nu}(p_{j,n}^*(d))\}\norm{p_{j,n}(d)-p_{j,n}^*(d)}_2^2\right]\\
            \leq &\frac{4\bar{u}^2}{\underline{d}^2}\mathbb{P}\left\{\sup_{d\in\Omega_{j,n}}\norm{p_{j,n}(d)-p_{j,n}^*(d)}_2>\nu\right\}\\
            \leq &\frac{4\bar{u}^2}{\underline{d}^2}C_1 \exp\left(-C_2\nu^2(n-j)\right)
        \end{aligned}
    \end{equation}
    Lemma \ref{Lemma: BoundedPrice} also implies that $\E\left[\sup_{d\in\Omega_{j,n}}\norm{p_{j,n}(d)-p^*_{j,n}(d)}_2^2\right]\leq\frac{4\bar{u}^2}{\underline{d}^2}$ if $n-j\leq N_1$. Together with 
    (\ref{Equation: DC_5}) and (\ref{Equation: DC_6}), the proof is complete.
\subsection{Proof of Lemma \ref{Lemma: UniformCylinder}}
We first provide a Lemma that shows a uniform Lipschitz continuous property of the population price. \begin{lem}\label{Lemma: LipschitzPrice}
    Under Assumptions \ref{Assumption: Asymptotic Regime}, \ref{Assumption: Nonstationarity} and \ref{Assumption: p_star}, there exists $L_d,\;\epsilon_d>0$ such that for all $P\in\mathcal{P}$ , when $n$ is large enough,
    \begin{equation*}
        \norm{p_{j,n}^*(d)-p_{j,n}^*(\delta_{j,n}(d_0))}_2\leq L_d\norm{d-\delta_{j,n}(d_0)}_2 \quad\forall d\in B_{\epsilon_d}(\delta_{j,n}(d_0))\;\forall 0\leq j<n
    \end{equation*}
\end{lem}
\begin{proof}
Fix $P\in\mathcal{P}$. Let $N_0$ be given by Proposition \ref{Proposition: Boundedness and Smoothness}(c) and fix $n\geq N_0$. Proposition \ref{Proposition: Boundedness and Smoothness}(c) and Lemma \ref{Lemma: Unique_Equal_Population_Price}(a) together show that $\delta_{j,n}(d_0)>\frac{1-\kappa}{2}\epsilon_p\cdot e$ for all $j$. Select $\epsilon_d \in (0,\min\{\frac{1-\kappa}{4}\epsilon_p,\frac{(1-\kappa)\underline{\sigma}\epsilon_p\epsilon^2}{64\bar u}\})$. Fix $d$ such that $\norm{d-\delta_{j,n}(d_0)}_2\leq \epsilon_d$
    \begin{equation}\label{eqn: Lip_Strong}
        \begin{aligned}
            0\leq &g_{j,n}^*(p_{j,n}^*(d),\delta_{j,n}(d_0))-g_{j,n}^*(p^*_{j,n}(\delta_{j,n}(d_0)),\delta_{j,n}(d_0))\\
            = & \left(g_{j,n}^*(p_{j,n}^*(d),\delta_{j,n}(d_0))- g_{j,n}^*(p_{j,n}^*(d),d)\right) + \left(g_{j,n}^*(p_{j,n}^*(d),d) - g_{j,n}^*(p_{j,n}^*(\delta_{j,n}(d_0)),d)\right)\\ 
            &+  \left(g_{j,n}^*(p_{j,n}^*(\delta_{j,n}(d_0)),d) - g_{j,n}^*(p^*_{j,n}(\delta_{j,n}(d_0)),\delta_{j,n}(d_0))\right)\\
            \leq & \abs{g_{j,n}^*(p_{j,n}^*(d),\delta_{j,n}(d_0))- g_{j,n}^*(p_{j,n}^*(d),d)} + \abs{g_{j,n}^*(p_{j,n}^*(\delta_{j,n}(d_0)),d) - g_{j,n}^*(p^*_{j,n}(\delta_{j,n}(d_0)),\delta_{j,n}(d_0))}\\
            \leq & \frac{8\bar u}{(1-\kappa)\epsilon_p}\norm{d-\delta_{j,n}(d_0)} \leq \frac{\underline{\sigma}\epsilon^2}{8}
        \end{aligned}
    \end{equation}
    where the last line is by Lemma \ref{Lemma: BoundedPrice-I}. By Proposition \ref{Proposition: Boundedness and Smoothness}(b), if $\norm{p-p_{j,n}^*(\delta_{j,n}(d_0))}_2 = \norm{p-p_{n}^*(d_0)}_2\in (\frac{\epsilon}{2},\epsilon)$, then for all $0\leq j < n$,
    \begin{equation*}
        g_{j,n}^*(p,\delta_{j,n}(d_0)) - g_{j,n}^*(p_{j,n}^*(\delta_{j,n}(d_0)),\delta_{j,n}(d_0))\geq\frac{\underline{\sigma}}{2}\norm{p-p_{j,n}^*(\delta_{j,n}(d_0))}^2_2>\frac{\underline{\sigma}\epsilon^2}{8}
    \end{equation*}
    Together with the level set of a convex function is convex, 
    \begin{equation*}
        \left\{p:g_{j,n}^*(p,\delta_{j,n}(d_0))\leq g_{j,n}^*(p_{j,n}^*(\delta_{j,n}(d_0)),\delta_{j,n}(d_0))+\frac{\underline{\sigma}\epsilon^2}{8}\right\}\subseteq \Omega_p
    \end{equation*}
    Then, (\ref{eqn: Lip_Strong}) shows that $p_{j,n}^*(d)\in\Omega_p$. By strong convexity again,
    \begin{equation*}
        \begin{aligned}
            &g_{j,n}^*(p_{j,n}^*(d),\delta_{j,n}(d_0))-g_{j,n}^*(p^*_{j,n}(\delta_{j,n}(d_0)),\delta_{j,n}(d_0))\geq \frac{\underline{\sigma}}{2}\norm{p_{j,n}^*(d)-p_{j,n}^*(\delta_{j,n}(d_0))}^2_2\\
            &g_{j,n}^*(p_{j,n}^*(\delta_{j,n}(d_0)),d)-g_{j,n}^*(p^*_{j,n}(d),d)\geq \frac{\underline{\sigma}}{2}\norm{p_{j,n}^*(d)-p_{j,n}^*(\delta_{j,n}(d_0))}^2_2
        \end{aligned}
    \end{equation*}
    Then, adding up the two inequalities
    \begin{equation*}
        \underline{\sigma}\norm{p_{j,n}^*(d)-p_{j,n}^*(\delta_{j,n}(d_0))}^2_2\leq (\delta_{j,n}(d_0)-d)^T(p_{j,n}^*(d)-p^*_{j,n}(\delta_{j,n}(d_0)))\leq \norm{d-\delta_{j,n}(d_0)}_2\norm{p_{j,n}^*(d)-p^*_{j,n}(\delta_{j,n}(d_0))}_2
    \end{equation*}
    Thus,
    \begin{equation*}
        \norm{p_{j,n}^*(d)-p_{j,n}^*(\delta_{j,n}(d_0))}_2\leq \frac{1}{\underline{\sigma}}\norm{d-\delta_{j,n}(d_0)}_2
    \end{equation*}
\end{proof}
To prove Lemma \ref{Lemma: UniformCylinder}, Fix $P\in\mathcal{P}$. Let $N_0$ be given by Proposition \ref{Proposition: Boundedness and Smoothness}(c) and fix $n\geq N_0$. We can assume $\epsilon_d$ given in Lemma \ref{Lemma: LipschitzPrice} is small enough such that $L_d\cdot\epsilon_d\leq \frac{\epsilon}{2}$. Then, Lemma \ref{Lemma: LipschitzPrice} implies that, for any $d\in B_{\epsilon_d}(\delta_{j,n}(d_0))$ and $0\leq j<n$, $p^*_{j,n}(d)\in B_{\frac{\epsilon}{2}}(p_n^*(d_0))\subseteq\Omega_p$. Then, by Proposition \ref{Proposition: Boundedness and Smoothness}(a) and (b), $\nabla g^*_{j,n}(p_{j,n}^*(d),d)=0$. Thus, $d = \bar\delta_{j,n}(p^*_{j,n}(d))$. Finally, by definition of $\Omega_{j,n}$, $\bar\delta_{j,n}(p^*_{j,n}(d))\in\Omega_{j,n}$. Thus, $d\in\Omega_{j,n}$. The proof is complete. 
\subsection{Proof of Lemma \ref{Lemma: TailProbMartingale}}
    Because $X_{j,n}$ and $Z_{j,n}$ are both measurable with respect to $\mathcal{F}_{j,n}$, $M_{j,n}$ is measurable with respect to $\mathcal{F}_{j,n}$. Also,
    \begin{equation*}
        \begin{aligned}
            \mathbb{E}\left[M_{j+1,n}\bigg\rvert\mathcal{F}_{j,n}\right]=&\mathbb{E}\left[X_{j+1,n} - \sum_{k=0}^{j}Z_{k,n}\bigg\rvert\mathcal{F}_{j,n}\right]\\
            =&\mathbb{E}\left[X_{j+1,n} \bigg\rvert\mathcal{F}_{j,n}\right] - \sum_{k=0}^{j}Z_{k,n}\\
            =& X_{j,n} + Z_{j,n} - \sum_{k=0}^{j}Z_{k,n} = X_{j,n} - \sum_{k=0}^{j-1}Z_{k,n} = M_{j,n}
        \end{aligned}
    \end{equation*}
    Thus, $\{M_{j,n}\}_{j=0}^n$ is a martingale adapted to $\{\mathcal{F}_{j,n}\}_{j=0}^n$. In addition,
    \begin{equation}\label{eqn: martingale_frozen}
        \begin{aligned}
            (M_{j+1,n} - M_{j,n})\mathbbm{1}\{j\geq \tau_n\} = &\left(X_{j+1,n} - X_{j,n} - Z_{j,n}\right)\mathbbm{1}\{j\geq \tau_n\}\\
            =&\left(X_{j+1,n}-\E[X_{j+1,n}\bigg\rvert\mathcal{F}_{j,n}]\right)\mathbbm{1}\{j\geq \tau_n\}\\
            =&X_{j+1,n}\mathbbm{1}\{j\geq \tau_n\} - \E[X_{j+1,n}\mathbbm{1}\{j\geq \tau_n\}\bigg\rvert\mathcal{F}_{j,n}]\\
            =& 2\alpha e \mathbbm{1}\{j\geq \tau_n\} - 2\alpha e \mathbbm{1}\{j\geq \tau_n\} = 0 
        \end{aligned}
    \end{equation}
    and let $M_{j,n}(i)$ and $X_{j,n}(i)$ be their $i$th entry, then
    \begin{equation*}
        (M_{j+1,n}(i) - M_{j,n}(i))\mathbbm{1}\{j< \tau_n\}\leq \abs{X_{j+1,n}(i)-X_{j,n}(i)}\mathbbm{1}\{j< \tau_n\} + \E\left[\abs{X_{j+1,n}(i)-X_{j,n}(i)}\mathbbm{1}\{j< \tau_n\}\bigg\rvert\mathcal{F}_{j,n}\right]
    \end{equation*}
    By Assumption \ref{Assumption: Nonstationarity} and Lemma \ref{Lemma: BoundedPrice}
    \begin{equation*}
        \begin{aligned}
            &\abs{X_{j+1,n}(i)-X_{j,n}(i)}\mathbbm{1}\{j< \tau_n\}\\
            \leq &\left(\abs{d_{j+1,n}(i) -d_{j,n}(i)} + \abs{\delta_{j+1,n}(d_0)(i)-\delta_{j,n}(d_0)(i)}\right)\mathbbm{1}\{j< \tau_n\}\\
            \leq &\left(\frac{{a_{j+1,n}(i)}+d_{j,n}(i)}{n-j-1}+\frac{\E[{a_{j+1,n}}(i)]+\delta_{j,n}(d_0)(i)}{n-j-1}\right)\mathbbm{1}\{t< \tau_T\}\leq  \frac{4\alpha+\epsilon_d}{n-j-1}
        \end{aligned}
    \end{equation*}
    Then, by Azuma's Inequality and Union Bound,
    \begin{equation*}
        \begin{aligned}
            P\left\{\norm{M_{j,n}}_2\geq\eta\right\}\leq 2m\exp\left(-\frac{\eta^2}{(8\alpha+2\epsilon_d)^2\sum_{k=1}^{j}\frac{1}{(n-k)^2}}\right)\leq 2m\exp\left(-\frac{\eta^2(n-j-1)}{(8\alpha+2\epsilon_d)^2}\right)
        \end{aligned}
    \end{equation*}
\subsection{Proof of Theorem \ref{Theorem: LogStoppingTime}}
We start from the following Lemma which can be viewed as a sufficient condition for Theorem \ref{Theorem: LogStoppingTime}. Let $\epsilon_{j,n}$'s and $\tau_n$ be defined in Theorem \ref{Theorem: LogStoppingTime}.
\begin{lem}\label{Lemma: O(1)StoppingTime}
    Under Assumptions \ref{Assumption: Asymptotic Regime}, \ref{Assumption: Single-sample}, \ref{Assumption: Nonstationarity}, \ref{Assumption: p_star}, if there exists $\eta>0$ and $\{\iota_n\}_{n=1}^{+\infty}$ such that for all $P\in\mathcal{P}$, $\epsilon_{n-\iota-1,n}\in (0,\epsilon_d)$ when $n$ is large enough, then
    \begin{equation*}
        \E[n-\tau_n]\leq 2m\left(\exp\left(\frac{\eta^2}{(16\alpha+4\epsilon_d)^2}\right) - 1\right)^{-1} + \iota_n + 1
    \end{equation*}
    for all $n$ large enough.
\end{lem}
\begin{proof}
    Fix $P\in\mathcal{P}$. Let $N_0$ be given in Proposition \ref{Proposition: Boundedness and Smoothness}(c). Fix $n\geq N_0$. If $n\leq \iota_n+1$, the result is trivial. Assume $n>\iota_n + 1$. By the construction of $\{\epsilon_{j,n}\}_{j=0}^{n-\iota_n-1}$, it is an increasing sequence. Thus, $0<\epsilon_{j,n}<\epsilon_d$ for all $0\leq j\leq n-\iota_n-1$. In addition, the construction of $\{\epsilon_{j,n}\}_{t=0}^{T-1}$ implies that
    \begin{equation*}
        \epsilon_{j,n}-\sum_{k=0}^{j-1}\max_{d\in B_{\epsilon_{k,n}}(\delta_{k,n}(d_0))}\norm{Z_n(k,d)}_2\geq\frac{\eta}{2}\quad\forall j=0,\cdots,n-\iota_n-1
    \end{equation*}
    Then, by (\ref{eqn: ZRelation}),
    \begin{equation*}
        P\left(\left\{\sum_{j=1}^{n-\iota_n-1}\left(\epsilon_{j,n}-\sum_{k=0}^{j-1}\norm{Z_{k,n}}_2\right)\mathbbm{1}\{\tau_n=j\}\geq\frac{\eta}{2}\right\}\;\;\bigcup\;\;\left\{\tau_n=n-\iota_n\right\}\right)=1
    \end{equation*}
    which implies that
    \begin{equation*}
        \mathbb{P}\left(\left\{\epsilon_{\tau_n,n}-\sum_{j=0}^{\tau_n-1}\norm{Z_{j,n}}_2\geq\frac{\eta}{2}\right\}\;\;\bigcup\;\;\left\{\tau_n=n-\iota_n\right\}\right)=1
    \end{equation*}
    Because $\norm{X_{\tau_n,n}}_2\geq\epsilon_{\tau_n,n}$ almost surely, then
    \begin{equation*}
        \begin{aligned}
            \norm{M_{\tau_n,n}}_2=\norm{X_{\tau_n,n}-\sum_{j=0}^{\tau_n-1}Z_{j,n}}_2\geq  \norm{X_{\tau_n,n}}_2-\norm{\sum_{j=0}^{\tau_n-1}Z_{j,n}}_2\geq \epsilon_{\tau_n,n}-\sum_{j=0}^{\tau_n-1}\norm{Z_{j,n}}_2
        \end{aligned}
    \end{equation*}
    Thus,
    \begin{equation}\label{eqn: martingale_universe}
        \mathbb{P}\left(\left\{\norm{M_{\tau_n,n}}\geq\frac{\eta}{2}\right\}\;\;\bigcup\;\;\left\{\tau_n=n-\iota_n\right\}\right)=1
    \end{equation}
    For any $1\leq j< n-\iota_n$, $\norm{X_{j,n}}_2\geq \epsilon_{j,n}$ if and only if $\tau_n\leq j$ because $0<\epsilon_{j,n}<\epsilon_d$. Then, by (\ref{eqn: martingale_frozen}), $\norm{X_{j,n}}_2\geq \epsilon_{j,n}$ implies that $M_{\tau_n,n}=M_{j,n}$. Together with (\ref{eqn: martingale_universe}) and Lemma \ref{Lemma: TailProbMartingale},
    \begin{equation*}
        \begin{aligned}
            \mathbb{P}\left\{\norm{X_{j,n}}_2\geq\epsilon_{j,n}\right\}\leq \mathbb{P}\left\{\norm{M_{j,n}}_2\geq\frac{\eta}{2}\right\}\leq 2m\exp\left(-\frac{\eta^2(n-j-1)}{(16\alpha+4\epsilon_d)^2}\right)
        \end{aligned}
    \end{equation*}
    Then,
    \begin{equation*}
        \begin{aligned}
            \E[n-\tau_n]=&\sum_{j=1}^{n-1}\mathbb{P}\{n-\tau_n\geq j\}\\
            \leq &\sum_{j=1}^{n-\iota_n-1}\mathbb{P}\left\{\norm{X_{j,n}}_2\geq\epsilon_{j,n}\right\}+\iota_n+1\\
            \leq & \sum_{j=1}^{n-\iota_n-1}2m\exp\left(-\frac{\eta^2(n-j-1)}{(16\alpha+4\epsilon_d)^2}\right)+\iota_n+1\leq 2m\left(\exp\left(\frac{\eta^2}{(16\alpha+4\epsilon_d)^2}\right) - 1\right)^{-1} + \iota_n + 1
        \end{aligned}
    \end{equation*}
\end{proof}
\noindent We are now ready to give the complete proof of Theorem \ref{Theorem: LogStoppingTime}. We start by showing the existence of $\eta$ and $\{\iota_n\}_{n=1}^{+\infty}$ that satisfy the hypothesis in Lemma \ref{Lemma: O(1)StoppingTime}. Fix $P\in\mathcal{P}$, let $N_0$ be given in Proposition \ref{Proposition: Boundedness and Smoothness}(c) and fix $n\geq N_0$. Given $d\in B_{\epsilon_d}(\delta_{j,n}(d_0))$ and $0\leq j< n$. Define
    \begin{equation*}
        \tilde p_{j,n}(d) = \arg\min_{p\geq 0} d^Tp+\frac{1}{n-j}\sum_{k=j+1}^n\left(\tilde u_{k,n}-{\tilde a_{k,n}}^Tp\right)^{+}
    \end{equation*}
    Then, by Lemma \ref{Lemma: Unique_Equal_Population_Price}
    \begin{equation*}
        \begin{aligned}
            &Z_n(j,d)\\
            =&\frac{1}{n-j-1}\E\left[(n-j)d-a_{j+1,n}\mathbbm{1}\left\{u_{j+1,n}>{a_{j+1,n}}^T\tilde p_{j,n}(d)\right\}\right]-\bar\delta_{j+1,n}(p_n^*(d_0))-\left(d-\bar\delta_{j,n}(p_n^*(d_0))\right)\\
            =&\frac{1}{n-j-1}\E\left[a_{j+1,n}\left(\mathbbm{1}\{u_{j+1,n}>{a_{j+1,n}}^Tp_{j,n}^*(d)\}-\mathbbm{1}\{u_{j+1,n}>{a_{j+1,n}}^T\tilde p_{j,n}(d)\}\right)\right]\\
            &+\frac{1}{n-j-1}\E\left[(n-j)d-a_{j+1,n}\mathbbm{1}\left\{u_{j+1,n}>{a_{j+1,n}}^Tp^*_{j,n}(d)\right\}\right]-\bar\delta_{j+1,n}(p_n^*(d_0))-\left(d-\bar\delta_{j,n}(p_n^*(d_0))\right)
        \end{aligned}
    \end{equation*}
    For simplicity of the notation, we define $A_{j,n}(p)\equiv a_{j,n}\mathbbm{1}\{u_{j,n}>{a_{j,n}}^Tp\}$. Together with the fact that $\nabla g_{j,n}^*(p_{j,n}^*(d);d)=0$ and $p_{n}^*(d_0)=p_{j,n}^*(\delta_{j,n}(d_0))$,
    \begin{equation}\label{eqn: Znjd_Decompose}
        \begin{aligned}
            &Z_n(j,d)\\
            =&\frac{\E\left[A_{j+1,n}(p_{j,n}^*(d))-A_{j+1,n}(\tilde p_{j,n}(d))\right]}{n-j-1}+\left(\bar\delta_{j+1,n}(p_{j,n}^*(d))-\bar\delta_{j+1,n}(p_{n}^*(d_0))\right) - \left(\bar\delta_{j,n}(p_{j,n}^*(d))-\bar\delta_{j,n}(p_{n}^*(d_0))\right)\\
            =&\frac{\bar\delta_{j+1,n}(p_{j,n}^*(d))-\bar\delta_{j+1,n}(p_{n}^*(d_0))}{n-j}-\frac{\E\left[A_{j+1,n}(p_{j,n}^*(d))-A_{j+1,n}(p_{n}^*(d_0))\right]}{n-j}+\frac{\E\left[A_{j+1,n}(p_{j,n}^*(d))-A_{j+1,n}(\tilde p_{j,n}(d))\right]}{n-j-1}
        \end{aligned}
    \end{equation}
    We here assume $a_{j,n}$'s have probability density function. The proof for the case that $a_{j,n}$'s have probability mass function is the same. We first upper bound the difference between the first two terms in the last line of (\ref{eqn: Znjd_Decompose}). For any $p_1$, $p_2$ in $\Omega_p$ and $k\geq j+2$,
    \begin{equation*}
        \E[A_{k,n}(p_1)-A_{k,n}(p_2)] = \int_{\Xi} a \int_{{a}^Tp_1}^{{a}^Tp_2} f\left(\frac{k}{n},a,s\right)ds v\left(\frac{k}{n},a\right)da
    \end{equation*}
    Recall $f\left(\frac{k}{n},a,\cdot\right)$ is the conditional density of $u_{k,n}$ given $a_{k,n} = a$ and $v\left(\frac{k}{n},a\right)$ is the density of $a_{k,n}$. Then, for any $k\geq j + 2$,
    \begin{equation*}
        \begin{aligned}
            &\norm{\E[A_{k,n}(p^*_{j,n}(d))-A_{k,n}(p^*_{n}(d_0))]-\E[A_{j+1,n}(p^*_{j,n}(d))-A_{j+1,n}(p^*_{n}(d_0))]}_2\\
            =&\norm{\int_{\Xi}a\left(v\left(\frac{k}{n},a\right)\int_{a^Tp_{j,n}^*(d)}^{a^T{p_n^*(d_0)}}f\left(\frac{k}{n},a,s\right)ds - v\left(\frac{j+1}{n},a\right)\int_{a^Tp_{j,n}^*(d)}^{a^T{p_n^*(d_0)}}f\left(\frac{j+1}{n},a,s\right)ds\right)da}_2\\
            \leq&\bar\mu\alpha^{m+2}\norm{p_{j,n}^*(d)-p_n^*(d_0)}_2\left(\sup_{a\in\Xi}\sup_{s\in [\underline{s},\bar s]}\abs{f\left(\frac{k}{n},a,s\right)-f\left(\frac{j+1}{n},a,s\right)}+\sup_{a\in\Xi}\abs{v\left(\frac{k}{n},a\right)-v\left(\frac{j+1}{n},a\right)}\right)
        \end{aligned}
    \end{equation*}
    where $\underline{s}=\inf_{a\in\Xi_a,p\in\Omega_p}a^Tp$ and $\bar{s}=\sup_{a\in\Xi_a,p\in\Omega_p}a^Tp$. By Lemma \ref{Lemma: Unique_Equal_Population_Price}, Proposition \ref{Proposition: Boundedness and Smoothness}(b) and by Mean Value Theorem, there exists $\tilde d$ in between $d$ and $\delta_{j,n}(d_0)$ such that
    \begin{equation*}
        \begin{aligned}
            \norm{p_{j,n}^*(d)-p_{n}^*(d_0)}_2 = \norm{p_{j,n}^*(d)-p_{j,n}^*(\delta_{j,n}(d_0))}_2\leq \norm{G^*_{j,n}(p_{j,n}^*(\tilde d))^{-1}}_2\norm{d-\delta_{j,n}(d_0)}_2\leq\frac{\norm{d-\delta_{j,n}(d_0)}_2}{\underline{\sigma}}
        \end{aligned}
    \end{equation*}
    Thus,
    \begin{equation*}
        \begin{aligned}
            &\norm{\frac{\bar\delta_{j+1,n}(p_{j,n}^*(d))-\bar\delta_{j+1,n}(p_{n}^*(d_0))}{n-j}-\frac{\E\left[A_{j+1,n}(p_{j,n}^*(d))-A_{j+1,n}(p_{n}^*(d_0))\right]}{n-j}}_2\\
            \leq & \frac{\alpha^{m+2}\sum_{k=j+2}^{n}\left(\sup_{a\in\Xi}\sup_{s\in [\underline{s},\bar s]}\abs{f\left(\frac{k}{n},a,s\right)-f\left(\frac{j+1}{n},a,s\right)}+\sup_{a\in\Xi}\abs{v\left(\frac{k}{n},a\right)-v\left(\frac{j+1}{n},a\right)}\right)}{\underline{\sigma}(n-j)(n-j-1)}\norm{d-\delta_{j,n}(d_0)}_2
        \end{aligned}
    \end{equation*}
    To simplify the notation, we define
    \begin{equation*}
        L_{j,n} = \frac{\bar\mu\alpha^{m+2}\sum_{k=j+2}^{n}\left(\sup_{a\in\Xi}\sup_{s\in [\underline{s},\bar s]}\abs{f\left(\frac{k}{n},a,s\right)-f\left(\frac{j+1}{n},a,s\right)}+\sup_{a\in\Xi}\abs{v\left(\frac{k}{n},a\right)-v\left(\frac{j+1}{n},a\right)}\right)}{\underline{\sigma}(n-j)(n-j-1)}
    \end{equation*}
    For the last term in (\ref{eqn: Znjd_Decompose}), let $\nu$ be given in Proposition \ref{Proposition: TailProbDualConv},
    \begin{equation*}
        \begin{aligned}
            \norm{\E\left[A_{j+1,n}(p^*_{j,n}(d))-A_{j+1,n}(\tilde p_{j,n}(d))\right]}_2\leq &\E\left[\norm{\E\left[A_{j+1,n}(p^*_{j,n}(d))-A_{j+1,n}(\tilde p_{j,n}(d))\bigg\rvert \tilde p_{j,n}(d)\right]}_2\mathbbm{1}\{\norm{\tilde p_{j,n}(d)-p^*_{j,n}(d)}_2\leq \nu\}\right]\\
            &+\alpha\mathbb{P}\left\{\norm{\tilde p_{j,n}(d)-p^*_{j,n}(d)}_2> \nu\right\}
        \end{aligned}
    \end{equation*}
    By Proposition \ref{Proposition: TailProbDualConv},
    \begin{equation*}
        \alpha\mathbb{P}\left\{\norm{\tilde p_{j,n}(d)-p^*_{j,n}(d)}_2> \nu\right\}\leq \alpha C_1 \exp\left(-C_2\nu^2(n-j)\right)
    \end{equation*}
    Because $(u_{j+1,n},a_{j+1,n})$ and $\tilde p_{j,n}(d)$ are independent by Assumption \ref{Assumption: Single-sample}(c), Proposition \ref{Proposition: Boundedness and Smoothness}(b) and Mean Value Theorem imply that,
    \begin{equation*}
        \begin{aligned}
            &\norm{\E\left[a_{j+1,n}\left(\mathbbm{1}\{u_{j+1,n}>{a_{j+1,n}}^Tp_{j,n}^*(d)\}-\mathbbm{1}\{u_{j+1,n}>{a_{j+1,n}}^T\tilde p_{j,n}(d)\}\right)\bigg\rvert \tilde p_{j,n}(d)\right]}_2\mathbbm{1}\{\norm{\tilde p_{j,n}(d)-p^*_{j,n}(d)}_2\leq \nu\}\\
            \leq &\bar\sigma\norm{p_{j,n}^*(d)-\tilde p_{j,n}(d)}_2
        \end{aligned}
    \end{equation*}
    Together with Theorem \ref{Theorem: DualConvergence}, there exists a constant $C_{3}$ that holds for all $P\in\mathcal{P}$ and does not depend on $j,\;n$ and $d$ such that
    \begin{equation*}
        \norm{\E\left[A_{j+1,n}(p^*_{j,n}(d))-A_{j+1,n}(\tilde p_{j,n}(d))\right]}_2\leq C_3\sqrt{\frac{\log n}{n - j -1}}
    \end{equation*}
    for all $n\geq N_{0}$. Define $e_{0,n} = \eta$, and
    \begin{equation*}
        e_{j+1,n} = e_{j,n} + L_{j,n}e_{j,n} + C_3\sqrt{\log n}(n-j-1)^{-\frac{3}{2}}\quad 0\leq j\leq n-2
    \end{equation*}
    Then, by Proposition \ref{Proposition: Boundedness and Smoothness}(d), 
    \begin{equation*}
        \begin{aligned}
            e_{n-\iota-1,n} = &\eta\prod_{j=0}^{n-\iota-2}(1+L_{j,n}) + C_3\sqrt{\log n}\sum_{j=1}^{n-\iota-1}\left(\prod_{q=j}^{n-\iota-2}(1+L_{q,n})(n-j)^{-\frac{3}{2}}\right)\\
            \leq & \eta e^{\frac{\mu\bar L\alpha^{m+2}}{\underline{\sigma}}} + C_3 e^{\frac{\mu\bar L\alpha^{m+2}}{\underline{\sigma}}}\sqrt{\log n}\sum_{j=\iota+1}^{n-1}j^{-\frac{3}{2}}
        \end{aligned}
    \end{equation*}
    In addition,
    \begin{equation*}
        \begin{aligned}
            \sum_{j=\iota+1}^{T-1}j^{-\frac{3}{2}}\leq\int_{\iota}^{+\infty}s^{-\frac{3}{2}}ds=\frac{2}{\sqrt{\iota}}
        \end{aligned}
    \end{equation*}
    Thus, there exists $\eta>0$ and constants $C_4$ hold for all $P\in\mathcal{P}$ such that, when $n\geq N_0$, $\iota_n\equiv \left\lceil C_4\log n\right\rceil\geq\left\lceil\frac{\alpha}{\underline{d}}\right\rceil$ and $e_{n-\iota_n-1,n}(\eta)<\epsilon_d$. In addition, $\epsilon_{0,n}(\eta)=e_{0,n}$. Suppose that $\epsilon_{j,n}\leq e_{j,n}$, then 
    \begin{equation*}
        \max_{d\in B_{\epsilon_{j,n}}(\delta_{j,n}(d_0))}\norm{Z(j,d)}_2\leq L_{j,n}e_{j,n} + C_3(n-j-1)^{-\frac{3}{2}} 
    \end{equation*}
    which implies that $\epsilon_{j+1,n}\leq e_{j+1,n}$. By induction, $\epsilon_{n-\iota_n-1,n}<\epsilon_d$. Thus, the hypothesis in Lemma \ref{Lemma: O(1)StoppingTime} is satisfied by $\eta$ and $\iota_n$'s, and $\tau_n$ defined by $\eta$ and $\iota_n$ satisfies (\ref{Definition: stopping_2}). Therefore, $\tau_n\leq \bar\tau_n$ almost surely, and
    \begin{equation*}
        \begin{aligned}
            \E[n-\bar\tau_n]\leq\E[n-\tau_n]\leq 2m\left(\exp\left(\frac{\eta^2}{(16\alpha+4\epsilon_d)^2}\right) - 1\right)^{-1} + \left\lceil C_4\log n\right\rceil + 1\equiv C_{Res}\log n
        \end{aligned}
    \end{equation*}
    when $n$ is large enough. The proof is complete.
\subsection{Proof of Lemma \ref{Lemma: MyopicRegret}}
We start from a Lemma that characterize the difference between $p_{j,n}(d_{j,n})$ and $p^*_{j-1,n}(d_{_{j-1,n}})$.
\begin{lem}\label{Lemma: Locatep_jn}
    Under Assumptions \ref{Assumption: Asymptotic Regime}, \ref{Assumption: Nonstationarity} , \ref{Assumption: p_star}, for any $P\in\mathcal{P}$, when $n$ is large enough,
    \begin{equation*}
        \norm{p_{j,n}(d_{j,n})-p_{j-1,n}^*(d_{j-1,n})}_2\mathbbm{1}\{j<\bar{\tau}_n\}\leq  \left(\sup_{d\in\Omega_{j,n}} \norm{p_{j,n}(d)-p_{j,n}^*(d)}_2+\frac{2\alpha}{\underline{\sigma}(n-j)}\right)\mathbbm{1}\{j<\bar{\tau}_n\}
    \end{equation*}
\end{lem}
\begin{proof}
    Fix $P\in\mathcal{P}$. Let $N_0$ be given in Proposition \ref{Proposition: Boundedness and Smoothness}(c), and fix $n\geq N_0$ and $1\leq j\leq n-1$. Then, define $EA_{j,n}(p) = \mathbb{E}\left[a_{j,n}\mathbbm{1}\{u_{j,n}>{a_{j,n}}^Tp\}\right]$, and by Lemma \ref{Lemma: Unique_Equal_Population_Price},
    \begin{equation*}
        p_{j-1,n}^*(d_{j-1,n})\mathbbm{1}\{j<\bar\tau_n\} = p_{j,n}^*\left(\frac{(n-j+1)d_{j-1,n}-EA_{j,n}(p_{j-1,n}^*(d_{j-1,n}))}{n-j}-d_{j,n}+d_{j,n}\right)\mathbbm{1}\{j<\bar\tau_n\}
    \end{equation*}
    Since on the event $\{j<\bar\tau_n\}$
    \begin{equation*}
        \begin{aligned}
            \norm{\frac{(n-j+1)d_{j-1,n}-EA_{j,n}(p_{j-1,n}^*(d_{j-1,n}))}{n-j}-d_{j,n}}_2\leq \norm{\frac{a_{j,n}\mathbbm{1}\{u_{j,n}>{a_{j,n}}^Tp_{j-1,n}^{\pi}\}-A_{j,n}(p_{j-1,n}^*(d_{j-1,n}))}{n-j}}_2\leq\frac{2\alpha}{n-j}
        \end{aligned}
    \end{equation*}    
    Lemma \ref{Lemma: Unique_Equal_Population_Price} and Mean Value Theorem implies that
    \begin{equation*}
        \begin{aligned}
            \norm{p^*_{j,n}(d_{j,n})-p_{j-1,n}^*(d_{j-1,n})}_2\mathbbm{1}\{j<\bar{\tau}_n\}\leq \frac{2\alpha}{\underline{\sigma}(n-j)}\mathbbm{1}\{j<\bar{\tau}_n\}
        \end{aligned}
    \end{equation*}
    Then, by triangular inequality,
    \begin{equation*}
        \begin{aligned}
            \norm{p_{j,n}(d_{j,n})-p_{j-1,n}^*(d_{j-1,n})}_2\mathbbm{1}\{j<\bar{\tau}_n\}\leq \left(\sup_{d\in\Omega_{j,n}} \norm{p_{j,n}(d)-p_{j,n}^*(d)}_2+\frac{2\alpha}{\underline{\sigma}(n-j)}\right)\mathbbm{1}\{j<\bar{\tau}_n\}
        \end{aligned}
    \end{equation*}    
\end{proof}
\noindent We can now give the complete proof of Lemma \ref{Lemma: MyopicRegret}. Fix $P\in\mathcal{P}$. Let $N_0$ be given in Proposition \ref{Proposition: Boundedness and Smoothness}(c), and fix $n\geq N_0$ and $1\leq j\leq n-1$. Define 
\begin{equation*}
    \begin{aligned}
        &\xi_{j,n,1} = \left(\sup_{d\in\Omega_{j,n}} \norm{p_{j,n}(d)-p_{j,n}^*(d)}_2+\frac{2\alpha}{\underline{\sigma}(n-j)}\right)\cdot e\\
        &\xi_{j-1,n,2} = \left(\sup_{d\in\Omega_{j-1,n}} \norm{\tilde p_{j-1,n}(d)-p_{j-1,n}^*(d)}_2\right)\cdot e
    \end{aligned}
\end{equation*}
where $e$ is the vector with all entries equal to $1$ and $\tilde p_{j,n}(d)$ is defined the same as in the proof of Theorem \ref{Theorem: LogStoppingTime}. Then, by Lemma \ref{Lemma: Locatep_jn}, on the event $\{j<\bar\tau_n\}$
\begin{equation*}
    \begin{aligned}
        {a_{j,n}}^Tp_{j,n}(d_{j,n}) = &{a_{j,n}}^Tp^*_{j-1,n}(d_{j-1,n}) + {a_{j,n}}^T\left(p_{j,n}(d_{j,n})-{a_{j,n}}^Tp^*_{j-1,n}(d_{j-1,n})\right)\\
        \geq & {a_{j,n}}^Tp^*_{j-1,n}(d_{j-1,n}) - {a_{j,n}}^T\xi_{j,n,1} = {a_{j,n}}^T\left(p^*_{j-1,n}(d_{j-1,n}) - \xi_{j,n,1}\right)
    \end{aligned}
\end{equation*}
In addition, on the event $\{j<\bar\tau_n\}$
\begin{equation*}
    \begin{aligned}
        {a_{j,n}}^Tp_{j-1,n}^{\pi} = &{a_{j,n}}^Tp^*_{j-1,n}(d_{j-1,n}) + {a_{j,n}}^T\left(p_{j-1,n}^{\pi}-{a_{j,n}}^Tp^*_{j-1,n}(d_{j-1,n})\right)\\
        \leq & {a_{j,n}}^Tp^*_{j-1,n}(d_{j-1,n}) + {a_{j,n}}^T\xi_{j,n,2} = {a_{j,n}}^T\left(p^*_{j-1,n}(d_{j-1,n}) + \xi_{j-1,n,2}\right)
    \end{aligned}
\end{equation*}
For simplicity of the notation, define $\Delta_{j,n}(p) = u_{j,n} - {a_{j,n}}^Tp$. Then,
\begin{equation}\label{eqn: Myopic}
    \begin{aligned}
        &\E\left[\Delta_{j,n}(p_{j,n}(d_{j,n}))\left(\mathbbm{1}\{\Delta_{j,n}(p_{j,n}(d_{j,n}))>0\}-\mathbbm{1}\{\Delta_{j,n}(p_{j-1,n}^{\pi})>0\}\right)\mathbbm{1}\{j<\bar\tau_n\}\right]\\
        =&\E\left[\Delta_{j,n}(p_{j,n}(d_{j,n}))\left(\mathbbm{1}\{\Delta_{j,n}(p_{j,n}(d_{j,n}))>0\}-\mathbbm{1}\{\Delta_{j,n}(p_{j-1,n}^{\pi})>0\}\right)\mathbbm{1}\{\Delta_{j,n}(p_{j,n}(d_{j,n}))>0\}\mathbbm{1}\{j<\bar\tau_n\}\right]\\
        &+\E\left[\Delta_{j,n}(p_{j,n}(d_{j,n}))\left(\mathbbm{1}\{\Delta_{j,n}(p_{j,n}(d_{j,n}))>0\}-\mathbbm{1}\{\Delta_{j,n}(p_{j-1,n}^{\pi})>0\}\right)\mathbbm{1}\{\Delta_{j,n}(p_{j,n}(d_{j,n}))<0\}\mathbbm{1}\{j<\bar\tau_n\}\right]
    \end{aligned}
\end{equation}
For the first term on the right-hand-side of (\ref{eqn: Myopic}),
\begin{equation*}
    \begin{aligned}
        &\E\left[\Delta_{j,n}(p_{j,n}(d_{j,n}))\left(\mathbbm{1}\{\Delta_{j,n}(p_{j,n}(d_{j,n}))>0\}-\mathbbm{1}\{\Delta_{j,n}(p_{j-1,n}^{\pi})>0\}\right)\mathbbm{1}\{\Delta_{j,n}(p_{j,n}(d_{j,n}))>0\}\mathbbm{1}\{j<\bar\tau_n\}\right]\\
        \leq & \E\left[\Delta_{j,n}(p^*_{j-1,n}(d_{j-1,n})-\xi_{j,n,1})\left(\mathbbm{1}\{\Delta_{j,n}(p^*_{j-1,n}(d_{j-1,n})-\xi_{j,n,1})>0\}\right.\right.\\
        &\left.\left.\qquad\qquad\qquad\qquad\qquad\qquad\qquad-\mathbbm{1}\{\Delta_{j,n}(p^*_{j-1,n}(d_{j-1,n}) + \xi_{j-1,n,2})>0\}\right)\mathbbm{1}\{j<\bar\tau_n\}\right]\\
        \leq &\E\left[{a_{j,n}}^T(\xi_{j,n,1}+\xi_{j,n,2})\left(\mathbbm{1}\{\Delta_{j,n}(p^*_{j-1,n}(d_{j-1,n})-\xi_{j,n,1})>0\}-\mathbbm{1}\{\Delta_{j,n}(p^*_{j-1,n}(d_{j-1,n}) + \xi_{j-1,n,2})>0\}\right)\mathbbm{1}\{j\leq\bar\tau_n\}\right]\\
        \leq & \E\left[\norm{\xi_{j,n,1}+\xi_{j-1,n,2}}_2\norm{\E\left[DA_{j,n}\bigg\rvert d_{j-1,n},\xi_{j,n,1},\xi_{j-1,n,2}\right]\mathbbm{1}\{j\leq \bar\tau_n\}}_2\right]
    \end{aligned}
\end{equation*}
where
\begin{equation*}
    DA_{j,n} = A_{j,n}(p^*_{j-1,n}(d_{j-1,n})-\xi_{j,n,1})-A_{j,n}(p^*_{j-1,n}(d_{j-1,n}) + \xi_{j-1,n,2})
\end{equation*}
and $A_{j,n}(p)$ is defined the same as the proof of Theorem \ref{Theorem: LogStoppingTime}. Let $\nu$ be given in Proposition \ref{Proposition: TailProbDualConv}. We can assume $\nu$ is small enough and that there exits $\tilde\iota$ not depending on $P$, $n$ and $j$ is large enough such that $\sup_{d\in\Omega_{j,n}} \norm{p_{j,n}(d)-p_{j,n}^*(d)}_2\leq \nu$ and $\sup_{d\in\Omega_{j-1,n}} \norm{\tilde p_{j-1,n}(d)-p_{j-1,n}^*(d)}_2\leq \nu$ imply that both $p^*_{j-1,n}(d_{j-1,n})-\xi_{j,n,1})$ and $p^*_{j-1,n}(d_{j-1,n})+\xi_{j-1,n,2})$ are in $\Omega_p$ when $j\leq \bar\tau_n\wedge n-\tilde\iota$. Let 
\begin{equation*}
    \mathcal{E}_{j,n} = \left\{\sup_{d\in\Omega_{j,n}} \norm{p_{j,n}(d)-p_{j,n}^*(d)}_2\leq \nu\right\}\cap \left\{\sup_{d\in\Omega_{j-1,n}} \norm{\tilde p_{j-1,n}(d)-p_{j-1,n}^*(d)}_2\leq \nu\right\}    
\end{equation*}
Then, by Proposition \ref{Proposition: Boundedness and Smoothness}(b) and Mean Value Theorem,
\begin{equation*}
    \begin{aligned}
        &\norm{\E\left[DA_{j,n}\bigg\rvert d_{j-1,n},\xi_{j,n,1},\xi_{j-1,n,2}\right]\mathbbm{1}_{\mathcal{E}_{j,n}}\mathbbm{1}\{j\leq \bar\tau_n\wedge n - \tilde\iota\}}_2\leq \bar\sigma \norm{\xi_{j,n,1}+\xi_{j-1,n,2}}_2
    \end{aligned}
\end{equation*}
Thus, if $j < n-\tilde\iota$, by Theorem \ref{Theorem: DualConvergence}
\begin{equation*}
    \begin{aligned}
        &\E\left[\norm{\xi_{j,n,1}+\xi_{j-1,n,2}}_2\norm{\E\left[DA_{j,n}\bigg\rvert d_{j-1,n},\xi_{j,n,1},\xi_{j-1,n,2}\right]\mathbbm{1}_{\mathcal{E}_{j,n}}\mathbbm{1}\{j\leq \bar\tau_n\}}_2\right]\\
        \leq &\E\left[\norm{\xi_{j,n,1}+\xi_{j-1,n,2}}_2^2\right]\leq\frac{C_4\log n}{n-j}
    \end{aligned}
\end{equation*}
where $C_4$ holds for all $P\in\mathcal{P}$, $n$ and $j$. In addition, Proposition \ref{Proposition: TailProbDualConv}, Lemma \ref{Lemma: BoundedPrice} and Assumption \ref{Assumption: Nonstationarity} show that
\begin{equation*}
    \E\left[\norm{\xi_{j,n,1}+\xi_{j-1,n,2}}_2\norm{\E\left[DA_{j,n}\bigg\rvert d_{j-1,n},\xi_{j,n,1},\xi_{j-1,n,2}\right]\mathbbm{1}_{\mathcal{E}_{j,n}^c}\mathbbm{1}\{j\leq \bar\tau_n\wedge n-\tilde\iota\}}_2\right] = o\left(\frac{1}{n-j}\right)
\end{equation*}
where the little-o term hold uniformly over $P\in\mathcal{P}$ and $n$. Then, we can assume $\tilde\iota$ is large enough such that there exists a constant $C_{Myopic}$ holds for all all $P\in\mathcal{P}$, $n$ and $j$ such that when $j < n-\tilde\iota$
\begin{equation*}
    \begin{aligned}
        \E\left[\Delta_{j,n}(p_{j,n}(d_{j,n}))\left(\mathbbm{1}\{\Delta_{j,n}(p_{j,n}(d_{j,n}))>0\}-\mathbbm{1}\{\Delta_{j,n}(p_{j-1,n}^{\pi})>0\}\right)\mathbbm{1}\{\Delta_{j,n}(p_{j,n}(d_{j,n}))>0\}\mathbbm{1}\{j<\bar\tau_n\}\right]\leq\frac{C_{Myopic} \log n}{n-j}
    \end{aligned}
\end{equation*}
Same analysis can be applied to the second term on the right-hand-side of (\ref{eqn: Myopic}). Thus, the proof is complete.
\end{document}